\documentclass[preprint,12pt]{elsarticle-arxiv}

\biboptions{sort}

\journal{Journal of Computer and System Sciences}

\usepackage{amssymb,amsmath,amsthm,xspace,wrapfig,appendix,pifont,bbding,amscd,natbib,nicefrac,ctable,trfsigns}

\newtheorem{theorem}{Theorem}[section]
\newtheorem{proposition}[theorem]{Proposition}
\newtheorem{lemma}[theorem]{Lemma}

\newtheorem{corollary}[theorem]{Corollary}

\newcommand{\ignore}[1]{}

\newenvironment{proof-sketch}{{\noindent\em Proof sketch.\ }}{\hfill{\Pisymbol{pzd}{113}}\vspace{0.1in}}

\newcommand{\R}{\mathbb{R}}

\newcommand{\x}{\mathbf{x}}
\newcommand{\W}{\mathbf{W}_{\mathrm{total}}}

\newcommand{\G}{\mathcal{G}}

\newcommand{\M}{\mathsf{M}}
\newcommand{\MS}{\mathsf{M}^{\mathrm{uns}}}
\newcommand{\D}{\mathsf{D}}

\newcommand{\NP}{\mathsf{NP}}
\newcommand{\LP}{\mathsf{LP}}
\newcommand{\ILP}{\mathsf{ILP}}
\newcommand{\SDP}{\mathsf{SDP}}

\newcommand{\TMIS}{3-$\mathsf{MIS}$}
\newcommand{\opt}{\mathsf{OPT}}

\newcommand{\apx}{\mathsf{APX}}
\newcommand{\DecProb}[1]{\mathsf{#1}}
\newcommand{\cS}{\mathcal{S}}

\newcommand{\eps}{\varepsilon}

\newcommand{\expect}{\operatorname{\mathbb{E}}}
\renewcommand{\epsilon}{\varepsilon}
\newcommand{\IE}{{\em i.e.}\xspace}
\newcommand{\EG}{{\em e.g.}\xspace}
\newcommand{\EA}{{\em et al.}\xspace}

\newcommand{\dhigh}{{\delta_{\mathrm{h}}}}
\newcommand{\dlow}{{\delta_{\mathrm{\ell}}}}

\newcommand{\dhighsq}{{\delta^2_{\mathrm{h}}}}

\newcommand{\mami}{\mbox{\small\sf max-min}}
\newcommand{\iin}{{{\mathrm{in}}}}
\newcommand{\out}{{{\mathrm{out}}}}

\begin{document}

\begin{frontmatter}



\title{On the Complexity of Newman's Community Finding Approach for Biological and Social Networks\tnoteref{bb1}}

\tnotetext[bb1]{Results in this paper were also presented at the ICALP 2011 workshop on Graph algorithms and Applications, Zurich, Switzerland, July 3, 2011.}


\author[auth1]{Bhaskar DasGupta\corref{label1}\fnref{label2}}

\address[auth1]{Department of Computer Science, University of Illinois at Chicago, Chicago, IL 60607. Email: {\tt dasgupta@cs.uic.edu}}

\fntext[label2]{Supported by DIMACS special focus on Computational and Mathematical Epidemiology. Research partially done 
while the author was on Sabbatical leave at DIMACS.}

\cortext[label1]{Corresponding author.}

\author[auth2]{Devendra Desai}

\address[auth2]{Department of Computer Science, Rutgers University, Piscataway, NJ 08854. \\ Email: {\tt devdesai@cs.rutgers.edu}}

\begin{abstract}
Given a graph of interactions,
a {\em module} (also called a {\em community} or {\em cluster}) is a subset of nodes
whose {\em fitness} is a function of the {\em statistical significance} of the pairwise interactions
of nodes in the module.
The topic of this paper is a {\em model-based} community finding approach,
commonly referred to as {\em modularity clustering},
that was originally proposed by Newman~\cite{LN08}
and has subsequently been extremely popular in practice (\EG, see~\cite{m1,m2,AK08,NG04,N06}).
Various heuristic methods are currently employed for finding the optimal solution.
However, as observed in~\cite{AK08}, the exact computational complexity of this approach is still largely
unknown.

To this end, we initiate a systematic study of the computational complexity of modularity clustering.
Due to the specific quadratic nature of the modularity function, it is necessary to study its value on {\em sparse}
graphs and {\em dense} graphs {\em separately}.
Our main results include a $(1+\eps)$-inapproximability for dense graphs and a logarithmic approximation
for sparse graphs. We make use of several combinatorial properties of modularity to get these results.
These are the first non-trivial approximability results
beyond the $\NP$-hardness results in~\cite{BDGGHNW07-1}.
\end{abstract}

\end{frontmatter}

\section{Introduction}

Many systems of interaction in biology and social science are modeled 
as a graph of pairwise interaction of entities~\cite{A05,AJB99}. 
An important problem for these types of graphs is to {\em partition}
the nodes into so-called ``communities'' or ``modules'' of ``statistically significant''
interactions. Such partitions facilitate studying interesting properties of
these graph in their applications, such as studying the behavioral patterns of
an individual in a societal context, and serve as important components 
in computational analysis of these graph.
In this paper we consider the {\em static} model of interaction in which
the network interconnections do not change over time.

Simplistic definitions of modules, such as {\em cliques}, unfortunately do not 
apply well in the context of biological and social networks and therefore alternative
definitions are most often used. In the ``model-based'' community finding 
approach, one first starts with an appropriate ``global null model'' $\G$ of a background random 
graph\footnote{Of course, any clustering measure that relies on a global null model
suffers from the drawback that each node can get attached to any other node of the graph; for 
another possible drawback see~\cite{FB07}. The purpose
of this paper is not to debate on the pros and cons of model-based clustering.}
and then attempts
to place nodes in the same module if their interaction patterns are significantly stronger than that
inferred from the null model. The null model $\G$ may provide, implicitly or explicitly, 
the probability $p_{i,j}$ of an edge between two nodes $v_i$ and $v_j$. 
As an illustration, suppose that our input is an edge-weighted graph with all weights being positive and
normalized between $0$ and $1$. Then, if $p_{i,j}$ differs significantly from $w_{i,j}$, the weight of the edge
between nodes $v_i$ and $v_j$, the edge may be considered to be {\em statistically significant}; thus, if 
$p_{i,j}\ll w_{i,j}$ then it is preferable that $v_i$ and $v_j$ should be placed in the same module whereas 
if $p_{i,j}\gg w_{i,j}$ then it is preferable that $v_i$ and $v_j$ should be placed in different modules.
The standard $\{+,-\}$-correlation clustering that appears in the computer science 
literature extensively~\cite{cor1,cor2,cor3} can be placed in the above model-based clustering framework in the
following manner: 
given the input graph $G$ with each edge labeled as $+$ or $-$, let $H$ be the 
graph consisting of all edges labeled $+$ in $G$, $p_{i,j}=0$ (resp. $p_{i,j}=1$) 
if the edge was labeled $+$ or missing (resp., labeled $-$), the modularity of an edge is 
$a_{i,j}-p_{i,j}$ where $a_{i,j}$ is the $(i,j)^{\rm th}$ entry in the adjacency matrix of $H$ and 
the total modularity is a function of individual modularities of edges as induced by the clustering.

In this paper, we investigate a model-based clustering approach 
originally introduced by Newman and subsequently studied by Newman and others in 
several papers~\cite{LN08,NG04,N06}.
The null model in this approach is dependent on the degree distribution
of the given graph. 
{\em Throughout the paper, by a set of communities (or clusters) we mean 
a partition $\cS$ of the nodes of the graph and, except in Section~\ref{digraph}, 
all graphs are undirected}. 

\subsection{The Basic Setup For Undirected Unweighted Graphs}

The basic setup for undirected unweighted graphs as described below can easily be generalized to the case of edge-weighted
undirected graphs (see Section~\ref{weighted}) and edge-weighted directed graphs (see Section~\ref{digraph}).
Let $G=(V,E)$ denote the given input graph with $n=|V|$ nodes and 
$m=|E|$ edges, let $d_v$ denote the degree of node $v\in V$, and let 
$A=\left[a_{u,v}\right]$ denote the {\em adjacency matrix} of $G$, \IE, $a_{u,v}=1$ if $\{u,v\}\in E$ and $a_{u,v}=0$ otherwise.
The null model $\G$ for modularity clustering is defined by the edge probability
function $p_{u,v}=\frac{d_u d_v}{2m}$ for $u,v\in V$ {\em with ${u=v}$ being allowed};
note that the null model provides a random network such that the {\em expected
degree} of a node $v$ is precisely $d_v$. 
Intuitively, if $a_{u,v}$ differs significantly from $p_{u,v}$ then 
the connection (or, the lack of it) is a significant deviation from the null
model. Based on this intuition, the {\em fitness} of the community formed by 
a subset of nodes $C\subseteq V$ is defined as\footnote{The $\nicefrac{1}{(2m)}$ factor is for {\em normalization purposes only} to 
make the {\em optimal} objective value to lie between $0$ and $1$.} 
\begin{equation}
\M(C)=
\frac{1}{2m} \left(
\sum_{u,v\in C} 
\left(
a_{u,v}-\frac{d_u d_v}{2m}
\right)
\right) 
\label{eq:11}
\end{equation}
Then, a partition $\cS=\big\{C_1,C_2,\ldots,C_k\big\}$ of $V$ has a {\em total modularity} of 
\begin{equation}
\M(\cS)=\sum_{C_i\in\cS}\M(C_i) 
\label{eq:1}
\end{equation}
Notice that each distinct pair of nodes $u$ and $v$ contribute {\em twice} to the inside term 
$a_{u,v}-\frac{d_u d_v}{2m}$ in Equation~\eqref{eq:11}. 
The goal is to find a partition (modular clustering) $\cS$ (with 
unspecified $k$) to {\em maximize} $\M(\cS)$. 
Note that by allowing $u$ and $v$ to be equal in the inside summation,
we provide a {\em negative weight to every node}.

Let $\displaystyle\opt=\max_{\cS} \M(\cS)$
denote the {\em optimal} modularity value.
It is easy to verify that $0\leq\opt<1$.

\subsection{Brief History of Modularity Clustering and Its Applications}

The modularity clustering approach is extremely popular both in the context of
biological networks~\cite{m1,m2} as well as social networks~\cite{AK08,LN08,NG04,N06}.
However, as observed in~\cite{AK08}, not much was known about the computational complexity aspect modularity
clustering beyond $\NP$-completeness for dense graphs, though various
heuristic methods have been proposed and empirically evaluated in 
publications such as~\cite{j1,j2,j3} via methods such as finding minimum weighted cuts.
For unweighted networks, it is known
that $\opt=0$ if $G$ is a clique, 
$\opt=1-\frac{1}{k}$ if $G$ is an union of $k$ {\em disjoint} cliques each with $n/k$ nodes,
computing $\opt$ is $\NP$-complete for sufficiently dense graphs\footnote{The reduction roughly requires 
$d_v=\Omega\left(\sqrt{n}\,\right)$ for every node $v$.} and 
the above-mentioned $\NP$-completeness result holds even if any solution is constrained to contain
no more than two clusters~\cite{BDGGHNW07-1}.

\subsection{Informal Summary of Our Results}

{\em Unless mentioned otherwise explicitly, all algorithmic results apply for edge-weighted graphs
and all hardness results apply for unweighted graphs}. 

\vspace*{0.1in}
\noindent
{\bf Hardness Results}
For {\em dense} graphs, namely for the complements of $3$-regular graphs, Theorem~\ref{maxsnp} in 
Section~\ref{inapprox} provides a $(1+\eps)$-inapproximability of the modularity clustering problem
irrespective of whether the number of clusters is
pre-specified or the algorithm is allowed to select the best number of 
clusters\footnote{The proof shows that $\eps$ is roughly $0.0006$.}.
The required approximation gap in our reduction is derived from the approximation gap of the maximum independent set
problem for $3$-regular graphs in~\cite{CC06}.
The intuition behind our inapproximability result is that, for the type of dense graphs that is considered in our reduction,
large-size cliques must be {\em properly} contained within the clusters.
However, the gap preservation calculations need to be done extremely accurately to avoid shrinking 
the inapproximability gap\footnote{For example, the inapproximability gap of Berman and Karpinski in~\cite{BK99} 
does not suffice for our purposes.}.

Lemma~\ref{ktotwo} in Section~\ref{fewclust} shows, using probabilistic arguments,  
that small number of clusters well-approximate 
the optimal modularity value; in particular, partitioning into {\em just two} clusters 
already achieves at least {\em half} of the optimum. Thus, it behooves to look at the complexity 
of the problem when we have at most two clusters, which we refer to as the 
{\em $2$-clustering problem}.
Theorem~\ref{nph} in Section~\ref{nphard} proves the $\NP$-completeness of the $2$-clustering problem for 
{\em sparse} graphs, namely for $d$-regular graphs with any fixed $d\geq 9$; the previous $\NP$-completeness result for this case
in~\cite{BDGGHNW07-1} required the degree of every node to be 
large (roughly $\Omega\left(\sqrt{n}\,\right)$ ). 
Notice that we cannot anymore use the idea of hiding a large-size clique since the graph does
not have any cliques of size more than $d$ and, for fixed $d$, one can indeed enumerate all these cliques in polynomial time.
Instead, our reduction is from the {\em graph bisection} problem for $4$-regular graphs. Intuitively, now an optimal solution for $2$-clustering 
is constrained to have exactly the same number of nodes in each community to avoid any local improvement.
The ideas in the reduction are motivated by the proof for this case in~\cite{BDGGHNW07-1},
but we have to do a more careful reduction and analysis to preserve both the low-degree and 
the regularity of the resulting graph.

\vspace*{0.1in}
\noindent
{\bf Approximation Algorithms}
We first consider the case of sparse graphs. We show in Section~\ref{lp-relax} that 
a natural linear programming relaxation of modularity clustering has a large
integrality gap, thereby ruling out this avenue for non-trivial 
approximations\footnote{Interestingly, the proof shows that $d$-regular expander graphs have 
small modularity values ($\approx \nicefrac{1}{\sqrt{d}}\,$).}.
Theorem~\ref{log-approx} in Section~\ref{log-regular}
provides a $O(\log d)$-approximation for most (unweighted) $d$-regular graph (\IE, with $d\leq \frac{n}{2\ln n}$),
and an approximation that is logarithmic in the {\em maximum weighted degree}
for weighted graphs provided maximum weighted degree\footnote{As noted in Section~\ref{weighted}, 
we normalize all the weights such that their sum is {\em exactly} twice the number of edges.}
is no more than {\em about} $\sqrt[5]{n}\,$.
It is easy to see that the modularity function is {\em neither monotone nor sub-modular}, 
thus we instead need to use semi-definite programming ($\SDP$) techniques for {\em maximizing quadratic forms}.
However, we face several technical hurdles in using 
$\SDP$-based approximation algorithms for quadratic forms in~\cite{CW04,AMMN05,AN06}: the coefficient 
matrix has {\em negative diagonal entries} and the lower bounds (hence the approximation ratios) 
in~\cite{CW04,AMMN05,AN06} depend on the number of nodes and not on the degree.
Thus, our proof proceeds in two steps. In the first step we obtain a lower bound on the optimal modularity value as a {\em function of the degree}
or {\em the maximum weighted degree} using an {\em explicit} graph
decomposition. In the second step, we show that the $\SDP$-based method for quadratic forms
can be used to obtain an approximation that is within a logarithmic factor of this lower
bound in spite of the negative diagonal entries.

For {\em locally-dense} weighted graphs (\IE, graphs in which every node has a weighted degree of $\Omega(n)\,$)
we observe in Section~\ref{RL} that one can get a solution within any {\em constant additive error} in polynomial time 
by a simple use of the {\em regularity lemma}. In view of our $\apx$-hardness result for dense graphs
described before, this is perhaps the best polynomial-time approximation one could hope for.

\vspace*{0.1in}
\noindent
{\bf Directed weighted Graphs}
In Section~\ref{digraph} we show that all the hardness and approximation results
for undirected weighted graphs can be extended to similar results for {\em directed} weighted graphs.

\vspace*{0.1in}
\noindent
{\bf Alternative Objectives and Null Models}
There are two natural objections to Newman's modularity clustering: 
{\em approximate solutions provably tend to produce many trivial (single-node) clusters} and 
{\em the background null model could be different}\footnote{The idea of using alternative null models has been explored 
before by some researchers~\cite{GGW07,KN09}; 
in particular, Karrer and Newman~\cite{KN09}
showed that the scale-free null model provided by linear preferential attachment 
do not provide a new null model.
However, the focus in all these results was mainly to {\em empirically}
compare null models using simple algorithms based on greedy approaches 
without provable approximation guarantees.}.
Motivated by these observations, we consider two variations of the original
modularity measure, one in which
the modularity of the network is the {\em minimum} (instead of sum) of the modularities of individual clusters
and the other in which 
the null model is the classical Erd\"{o}s-R\'{e}nyi random graph.
Our results show that the minimum objective provides similar optimal
modularity values as the original sum objective without allowing small clusters, and the 
Erd\"{o}s-R\'{e}nyi random graph null model is equivalent to Newman's modularity clustering in an {\em appropriately
defined} regular graph.

\subsection{Comments on Our Results}

\vspace*{0.1in}
\noindent
{\bf Relationships to previous approximation algorithms for quadratic forms}
The special case of partitioning the nodes into {\em two clusters only} can be written down as 
maximizing a quadratic form. However, none of the existing approximability results for quadratic forms apply directly
to our case. In particular, the $O(\log n)$-approximation in~\cite{CW04,AMMN05}
is not applicable since the diagonal entries of the resulting constraint matrix 
are negative\footnote{The negative diagonal entries {\em are crucial} in the modularity measure~\cite{personal,AK08}. Moreover, they
could be small or large depending on the graph, thus it is not possible to specify a priori
bound on them.}, results such as in~\cite{HK04} do not apply
since the constraint matrix is {\em not} necessarily a positive semi-definite matrix and 
the $O(1)$-approximations of~\cite{AN06} via Grothendieck's inequality do not apply since 
the quadratic form does {\em not} induce a bipartition of variables.

\vspace*{0.1in}
\noindent
{\bf Possibility of logarithmic approximation without degree constraints}
Our logarithmic approximations require some bound on the maximum degree of the given graph.
A natural question is of course if such degree bounds can be removed. Two 
observations regarding this are relevant:

\vspace*{0.1in}
\noindent
\SixStar$\,$
A technical difficulty that arises for this purpose is from the fact that the modularity value 
can be precisely $0$ (such as when the given graph is $K_n$, $K_{n,n}$ or a graph 
obtained from $K_n$ by removing $\mathrm{polylog}(n)$ edges) or 
arbitrarily close to $0$ (such as when the given graph is the complement of 
small degree graph). Thus, at the very least, a non-trivial approximation without
such degree bounds would require an efficient 
polynomial-time computable characterization of the topology of graphs whose modularity values can be arbitrarily small
together with a special algorithmic approach to handle these graphs;
approaches using quadratic forms or the regularity lemma do not suffice in this respect.

\vspace*{0.1in}
\noindent
\SixStar$\,$
The negative weights of the nodes start playing a more crucial role in the value of 
modularity when it is close to $0$. 
As observed by other researchers before, negative diagonal entries in the coefficient 
matrix of the objective that shifts the objective value close to $0$
are sometimes difficult for approximate.

\vspace*{0.1in}
\noindent
{\bf Relationships to other clustering or partitioning methods}
Modularity clustering can be defined by several equivalent equations, which may {\em seem to suggest}
at a first glance that combinatorially the problem may be either similar to
(via Equations~\eqref{eq:11}~and~\eqref{eq:1}$\,$)
some form of {\em correlation clustering}, or (via Equation~\eqref{eq:3}$\,$) similar to 
{\em graph bisection} (for two clusters), or similar to {\em minimum $\ell$-way cut}/{\em clique-partition} 
type of problem (for arbitrary number of clusters, depending on whether the graph is unweighted or weighted), or 
similar to (via Lemma~\ref{subgraph-selection}) some type of {\em dense subgraph} problem. 
However, our results show both similarities and differences between modularity clustering and these problems.
For example, our 
hardness result for dense graphs should be contrasted with other partitioning problems of similar
nature, such as MAX-CUT, graph bisection, graph separation, minimum $\ell$-way cut
and some versions of correlation clustering,
for which one can design a PTAS (\EG, see~\cite{AKK99,cor1,FK96}). 

\section{Basic Results on Partitioning into Fewer Clusters}
\label{fewclust}

In this section we show bounds on $\opt$ as well as some useful properties of 
the solution if we restrict the number of clusters to some pre-specified value $k$; we will refer 
to this as the {\em $k$-clustering} problem. The objective function $\M(\cS)$ can be equivalently represented (via
algebraic manipulation as observed in~\cite{BDGGHNW07-1,LN08,NG04,N06}) as follows.
Let $m_i$ denote the number of edges whose both endpoints are in the cluster $C_i$, 
$m_{ij}$ denote the number of edges one of whose endpoints is in $C_i$ and the other in $C_j$ 
and $D_i = \sum_{v\in C_i} d_v$ denote the sum of degrees of nodes in cluster $C_i$. Then, 
\begin{equation}
\M(\cS)=
\sum_{C_i\in\cS}\, 
\left(
\frac{m_i}{m} - \left( \frac{D_i}{2m} \right)^2 
\right) 
\label{eq:2}
\end{equation}
Since 
$\sum_{v\in V} \left(a_{u,v}-\frac{d_u d_v}{2m}\right)=0$ for any $u\in V$,
we can alternatively express $\M(C)$ as 
\begin{equation}
\M(C)=
\frac{1}{2m} \left(
\sum_{u\in\, C,\,v\not\in\, C} 
\left(
\frac{d_u d_v}{2m}-a_{u,v}
\right)
\right) 
\label{eq:hh}
\end{equation}
This, along with Equation~\eqref{eq:2}, this gives us the following third equation of modularity
(note that now each pair of clusters contributes to the sum in Equation~\eqref{eq:3} {\em exactly once}):
\begin{equation}
\M(\cS)= \sum_{C_i,\,C_j\,\colon\,i\,<j}{\left(\frac{D_i D_j}{2m^2} - \frac{m_{ij}}{m}\right)}
\label{eq:3}
\end{equation}
Let $\opt_k$ denote the modularity value of an optimal clustering when one is allowed 
at most $k$ clusters. 

The following two lemmas make use of the alternative formulations described above.
The first lemma asserts, via a probabilistic argument,
that the optimal value does not go down by too much in our restricted setting. 

\begin{lemma}\label{ktotwo}
For any $k\geq 1$, $\left(1-\frac{1}{k}\right)\,\opt\leq\opt_k\leq 1-\frac{1}{k}$.
\end{lemma}

\begin{proof}
The inequality $\opt_k\leq 1-\frac{1}{k}$ can be proved as follows.
For any clustering $\cS$ with at most $k$ clusters, Equation~\eqref{eq:2} gives
$\M(\cS) = \sum_{i=1}^k \frac{m_i}{m} - \sum_{i=1}^k \left(\frac{D_i}{2m}\right)^2$.
The first sum in this equation is upper-bounded by $1$. Using Cauchy-Schwarz inequality, we get $k\sum_{i=1}^k D_i^2\geq\left(\sum_{i=1}^k D_i\right)^2$, 
giving a lower-bound of $\nicefrac{1}{k}$ for the second sum. 

The inequality $\left(1-\frac{1}{k}\right)\opt\leq\opt_k$ can be proved as follows.
For $k=1$, the statement is trivially true. Now consider $k>1$. We will make use of Equation~\eqref{eq:3} for modularity values. 
Suppose that our optimal clustering $\cS$ has more than $k$ clusters. Denote each term in the summation of 
Equation~\eqref{eq:3} by $\M_{ij}$, \IE, $\M_{ij}=\frac{D_i D_j}{2m^2} - \frac{m_{ij}}{m}$; thus
$\opt = \M(\cS)= \sum_{i<j}{\M_{ij}}$.
We can randomly assign each of the clusters to one of $k$ superclusters. Let $I_{ij}$ be the indicator random variable of the event $C_i$ and $C_j$ are in different clusters and 
let $\cS_k$ denote the random $k$-clustering. It is easy to see that any pair $C_i$ and $C_j$ will contribute $\M_{ij}$ to the final 
clustering if and only if they are not in the same supercluster. Therefore, $\M(\cS_k) = \sum_{i<j}I_{ij}\M_{ij}$. Thus we get
$
\opt_k\geq\expect[\M(\cS_k)]=\sum_{i<j}{\expect[I_{ij}]\M_{ij}}=\sum_{i<j}{\left(1-\frac{1}{k}\right)\M_{ij}} = \left(1-\frac{1}{k}\right)\opt$. 
\end{proof}

The next lemma shows that the $2$-clustering problem can also be alternatively viewed as 
a special kind of ``subgraph selection'' problem. 

\begin{lemma}\label{subgraph-selection}
Let $V_1$ and $V_2$ be any partition of $V$. Then, $\M(V_1)=\M(V_2)$.
\end{lemma}

\begin{proof}
Remember that, for any node $u$, 
$\sum_{v\in V}\left(a_{u,v}-\frac{d_u d_v}{2m}\right)=0$.
Thus, 
\begin{gather*}
0=\sum_{u\in V_1}\sum_{v\in V}\left(a_{u,v}-\frac{d_u d_v}{2m}\right)=\M(V_1)+\sum_{u\in V_1}\sum_{v\in V_2}\left(a_{u,v}-\frac{d_u d_v}{2m}\right)
\\
0=\sum_{u\in V_2}\sum_{v\in V}\left(a_{u,v}-\frac{d_u d_v}{2m}\right)=\M(V_2)+\sum_{u\in V_2}\sum_{v\in V_1}\left(a_{u,v}-\frac{d_u d_v}{2m}\right)
\end{gather*}
and therefore $\M(V_1)=\M(V_2)$.
\end{proof}

\section{Results for Dense Graphs}

\subsection{$\apx$-hardness}
\label{inapprox}

This hardness result may be contrasted with the results in Section~\ref{RL}
where we show that the modularity value can be approximated to within any {\em constant
additive error} for dense graphs using the regularity lemma.
However, the $\apx$-hard instances here have modularity values that
are very close to $0$ (around $\nicefrac{1}{n}$), thus the constant additive
error provides no guarantee on the approximation ratio.

\begin{theorem}\label{maxsnp}
It is $\NP$-hard to approximate the $k$-clustering problem, for any $k$,  
on $(n-4)$-regular graphs within a factor of $1+\eps$ for some constant $\eps>0$.
\end{theorem}

\begin{proof}
We reduce the maximum-cardinality independent set problem for $3$-regular
graphs (\TMIS) to our problem.
An instance of \TMIS\ consists of a $3$-regular graph $H=(V,E)$,  
and the goal is to find a {\em maximum cardinality} subset of nodes 
$V'\subset V$ such that every pair of nodes $u$ and $v$ in $V'$ is {\em independent}, \IE, $\{u,v\}\not\in E$. 
For notational convenience, let $\dlow=\nicefrac{94}{194}$ and $\dhigh=\nicefrac{95}{194}$. The following inapproximability result is known for \TMIS.

\begin{theorem}{\rm~\hspace*{-0.1in}\cite{CC06}}\label{bk}
For any language $L$ in $\NP$,
there exists a polynomial-time reduction
such that given an instance
$I$ of $L$ produces an instance of $H$ of {\rm \TMIS} with $n$ nodes such that:
\begin{itemize}
\item
if $I\in L$ then $H$ has a maximum independent set of cardinality at least $\dhigh n$; 

\item
if $I\not\in L$ then every maximum independent set of $H$ is of cardinality at most $\dlow\,n$. 
\end{itemize}
\end{theorem}

\noindent
We start with an instance $I$ of $L$ and translate it to an instance $H$ 
of \TMIS\ as described in Theorem~\ref{bk}; we refer to such an instance of \TMIS\
as a ``hard'' instance. Given a hard instance $H=(V,F)$ of \TMIS\ with $|V|=n$ nodes and $|F|=\frac{3\,n}{2}$ edges such that 
a maximum independent set is of size either at most $\dlow\, n$ or at least $\dhigh n$, 
consider the complement $\overline{\rule[8pt]{0pt}{2pt}H}=(V,\overline{\rule[8pt]{0pt}{2pt}F})$ of $H$, \IE,
the graph with $\overline{\rule[8pt]{0pt}{2pt}F}=\{\,\{u,v\}\,|\,u,v\in V,\,u\neq v\}\setminus F$.
Since $H$ is $3$-regular, $\overline{\rule[8pt]{0pt}{2pt}H}$ is $(n-4)$-regular.
The input to our $2$-clustering problem is this graph $\overline{\rule[8pt]{0pt}{2pt}H}$.
For notational uniformity, we will denote the graph $\overline{\rule[8pt]{0pt}{2pt}H}$
by $G=(V,E)$ with $E=\overline{\rule[8pt]{0pt}{2pt}F}$.
Note that $V'\subset V$ is an independent set of $H$ if and only if 
$V'$ is a clique in $G$.
Let $\mathbf{\Psi}$ and $\opt$ denote the size of a maximum independent set of $H$ and 
the optimal modularity value of $G$, respectively.
We prove our claim by showing the following:
\begin{description}
\item[\hspace*{0.5in}(completeness)]
If $\mathbf{\Psi}\geq\dhigh n$
then 
$\opt\geq\dfrac{2(4\dhighsq-\dhigh)}{(n-4)}>\dfrac{0.9388}{n-4}$.

\item[\hspace*{0.75in}(soundness)]
If $\mathbf{\Psi}\leq\dlow\, n$
then 
$\opt\leq \dfrac{4\dlow-1}{n-4}<\dfrac{0.9382}{n-4}$.
\end{description}
For any subset $\emptyset\subset V'\subset V$ of nodes in $G$, 
let $m_{V'}$ be the number of edges in $G$ with both end-points in $V'$ and $D_{V'}$ be the sum of degrees 
of nodes in $V'$ in the graph $G$, \IE, $D_{V'}=\sum_{v\in V'}d_{v}$.

\subsubsection{Proof of Completeness {\rm (}$\mathbf{\Psi}\geq\dhigh n${\rm )}}

\begin{lemma}\label{easydir}
If $\mathbf{\Psi}\geq\dhigh n$
then 
$\opt\geq\dfrac{2(4\dhighsq-\dhigh)}{(n-4)}$.
\end{lemma}

\begin{proof}
Suppose $H$ has a has an independent set $V'$ with $|V'|=t\,n$ for some $t\geq\dhigh$. 
Since $V'$ is a clique of $G$, it follows that $2m_{V'}=tn(tn-1)$ and $D_{V'}=tn(n-4)$. Consider the solution $\cS=\big\{V',V\setminus V'\big\}$
of $2$-clustering on $G$. Using Lemma~\ref{subgraph-selection} and Equation~\eqref{eq:2} we get
\begin{gather*}
\M(\cS)=2\,\M(V')
=2\left(\frac{m_{V'}}{m}-\left(\frac{D_{V'}}{2m}\right)^2\right)
\\
=\frac{2\,t\,n\,(t\,n-1)}{n(n-4)}-2\,t^2
=\frac{2(4\,t^2-t)}{n-4}
\geq\frac{2(4\dhigh^2-\dhigh)}{n-4}
\qedhere
\end{gather*}
\end{proof}

\subsubsection{Proof of Soundness {\rm (}$\mathbf{\Psi}\leq\dlow\,n${\em )}}

\noindent
{\bf Case I: when an optimal solution has exactly $2$ clusters}.

\vspace*{0.1in}
\noindent
Suppose that the optimal solution is 
$\cS=\big\{V',V\setminus V'\big\}$ of $2$-clustering on $G$ with $|V'|=t\,n$ and $0<t\leq \nicefrac{1}{2}$.

\begin{lemma}\label{exact-count}
Let $\alpha\,n$ be the size (number of nodes) of a largest size clique in 
the node-induced subgraph $G'=(V',E')$ where $E'=(V' \times V')\cap E$.  
Then, $\M(V')\leq\dfrac{4t^2+2\alpha-3t}{n-4}$.
\end{lemma}

\begin{proof}
Since the size of the largest clique in $G'$ is $\alpha\, n$, for each of the remaining $(t-\alpha)\,n$ nodes, 
they will not be connected to at least one node inside the clique. Hence, using Equation~\eqref{eq:2}, we get
\[
\M(V') = \frac{m_{V'}}{m}-\left(\frac{D_{V'}}{2m}\right)^2
\leq \frac{\frac{t\,n(t\,n-1)}{2}-(t-\alpha)\,n}{\frac{n(n-4)}{2}}-t^2
=\frac{4\,t^2+2\,\alpha-3\,t}{n-4}
\qedhere
\]
\end{proof}

\begin{lemma}\label{less-than-half}
$\M(V')\leq\dfrac{2\dlow-\frac{1}{2}}{n-4}$.
\end{lemma}

\begin{proof}
Using the previous lemma and the facts that $\alpha\leq\min\big\{t,\dlow\big\}$ and $t\leq \nicefrac{1}{2}$,
we have two cases:

\vspace*{0.1in}
\noindent
{\bf Case 1:} $\pmb{t>\dlow}$. 
Then $\M(V')\leq \frac{4t^2+2\alpha-3t}{n-4}$. The function $f(t)=4t^2-3t$ is increasing in the range $\left(\left.\dlow,\nicefrac{1}{2}\right]\right.$ 
since $\dlow>3/8$ and $\frac{\partial f}{\partial\, t}=8t-3>0$ if $t>3/8$.
Thus, $\max_{\dlow<t\leq \nicefrac{1}{2}} f(t)=f\left(\nicefrac{1}{2}\right)=-\nicefrac{1}{2}$, 
and thus $\M(V')\leq\frac{2\alpha-\frac{1}{2}}{n-4}\leq\frac{2\dlow-\frac{1}{2}}{n-4}$.

\vspace*{0.1in}
\noindent
{\bf Case 2:} $\pmb{t\leq\dlow}$.
Since $\alpha\leq t$ and $4t^2+2\alpha-3t$ is an increasing function of $\alpha$, we have 
$\M(V')\leq \frac{4t^2+2t-3t}{n-4}=\frac{4t^2-t}{n-4}$.
The function $f(t)=4t^2-t$ satisfies $f(0)=0$ and 
\[
\frac{\partial f}{\partial t}=8t-1
\left\{
\begin{array}{ll}
<0 & \mbox{if $t<1/8$} \\
>0 & \mbox{if $1/8<t\leq\dlow$} \\
\end{array}
\right.
\]
Thus, $\max_{0<t\leq\dlow}f(t)=f(\dlow)$ and we have 
$
\M(V')
\leq\dfrac{4\dlow^2-\dlow}{n-4}
\leq\dfrac{2\dlow-\frac{1}{2}}{n-4}
$.
\end{proof}

\noindent
Finally, using Lemma~\ref{subgraph-selection},
$\M(\cS) = 2\,\M(V') \leq \dfrac{4\dlow - 1}{n-4}$,$\,$ completing the soundness proof for this case.

\vspace*{0.1in}
\noindent
{\bf Case II: when an optimal solution has more than $2$ clusters}.

\vspace*{0.1in}
\noindent
For convenience of calculations, we would like to drop the $\frac{1}{2m}$ scaling term from Equation~\eqref{eq:11}.
To this end, we define $\MS(C)=n(n-4)\,\M(C)$.
Let $\cS=\big\{V_1,V_2,\dots,V_{m+1}\big\}$ be an optimal solution of modularity 
clustering that uses a {\em minimum} $m>1$ number of clusters.
Let $|V_i|=t_i\,n$, and suppose that $\emptyset\subset V_i'\subseteq V_i$ is
a largest clique of size $\alpha_i\,n$ in the graph $(V_i,\,(V_i \times V_i)\cap E)$.
Note that $0<\alpha_i\leq\min\big\{\,t_i,\dlow\big\}$ for all $1\leq i\leq m+1$, 
$\sum_{i=1}^{m+1}t_i=1$ and we need to show that $\MS(\cS)\leq\big(4\dlow-1\big)\,n$.
Let $\widehat{V_i}$ denote $V\setminus V_i$.

\begin{lemma}\label{small-component}
$\MS(V_i)\leq (4t_i^2-t_i)  n$.
\end{lemma}

\begin{proof}
$\MS(V_i)$ is maximized when the nodes in $V_i$ form a clique.
Thus,
\[
\MS(V_i)
\leq
\left(\frac{4}{n}-1\right)  (t_i n) + \left(\frac{4}{n}\right)\left(t_in-1\right)(t_i n)
=\left(4t_i^2-t_i\right)  n
\qedhere
\]
\end{proof}

\begin{corollary}\label{sm-cor}
If $|V_i|\leq \nicefrac{n}{4}$ then $\MS(V_i)\leq 0$. If $|V_i|=\left(\frac{1}{4}+\delta\right)n>\nicefrac{n}{4}$ 
then $\MS(V_i)\leq\left(4\delta^2+\delta\right)n$.
\end{corollary}

\begin{lemma}\label{complement}
Suppose that $t_i=\frac{1}{2}+\delta>\frac{1}{2}$ for some $0<\delta<\nicefrac{1}{2}$
and $\widehat{\alpha_i}$ is the size of a largest clique
in $(\widehat{V_i},(\widehat{V_i} \times \widehat{V_i})\cap E)$.
Then,
\[
\MS(V_i)\leq\left(4\delta^2-\delta-\frac{1}{2}+2\widehat{\alpha_i}\right)\,n\leq\left(2\dlow - \frac{1}{2}\right)\,n
\]
\end{lemma}

\begin{proof}
Note that $\left|\widehat{V_i}\right|=\frac{1}{2}-\delta<\nicefrac{1}{2}$. 
Then by Lemma~\ref{subgraph-selection},
\begin{gather*}
\MS(V_i)=\MS(\widehat{V_i})
\leq\left(4 \left(\frac{1}{2}-\delta\right)^2+2\widehat{\alpha_i}-3\left(\frac{1}{2}-\delta\right)\right)  n 
\\
=
\left(4\delta^2-\delta-\frac{1}{2}+2\widehat{\alpha_i}\right)  n
\end{gather*}
where the inequality follows from Lemma~\ref{exact-count} if we replace $V_i$ by $\widehat{V_i}$. Since $t_i\geq\nicefrac{1}{2}$, we 
have 
\[
4\delta^2-\delta-\frac{1}{2}+2\widehat{\alpha_i}=4t_i^2-5t_i+1-2\widehat{\alpha_i}\leq 4t_i^2+2  \widehat{\alpha_i}-3  t_i
\]
Since $\widehat{\alpha_i}\leq\dlow<t_i$, the arguments in Lemma~\ref{less-than-half} can be directly applied on 
$4t_i^2+2  \widehat{\alpha_i}-3  t_i$
to show that  
$\left(4\delta^2-\delta-\frac{1}{2}+2\widehat{\alpha_i}\right) n \leq \left(2\dlow - \frac{1}{2} \right)n$.
\end{proof}

Let us call a cluster $V_i$ a {\em giant component} if $t_i>\dlow$. Note that 
since $3\,\dlow>1$, we can have {\em at most two} giant components.
We have therefore three cases depending on the number of giant components.

\vspace*{0.1in}
\noindent
{\bf Case (i): $\cS$ has no giant components} 
Note that $\cS$ can have at most three clusters containing {\em strictly} more than $\nicefrac{n}{4}$ nodes.

If $\cS$ contains no such cluster then by Corollary~\ref{sm-cor} $\MS(\cS)\leq 0$.

If $\cS$ contains exactly one such cluster, say $V_1$, then 
$\MS(\cS)\leq\MS(V_1)\leq \left(2\dlow-\frac{1}{2}\right)\,n<(4\dlow-1)\,n$ by Lemma~\ref{less-than-half} (if $t_i\leq\nicefrac{1}{2}$) or Lemma~\ref{complement} (if $t_i>\nicefrac{1}{2}$).

If $\cS$ contains exactly two such clusters, say $V_1$ and $V_2$, then again
$\MS(\cS)\leq\MS(V_1)+\MS(V_2)\leq 2\left(2\dlow-\frac{1}{2}\right)\,n=\big(4\dlow-1\big)\,n\,$ by Lemma~\ref{less-than-half} and Lemma~\ref{complement}.

Otherwise, suppose that $\cS$ contains {\em exactly three} such clusters, say $V_1$, $V_2$ and $V_3$.
Let $t_i=\frac{1}{4}+\delta_i$ for $i=1,2,3$.
Then, $0<\delta_1+\delta_2+\delta_3<\nicefrac{1}{4}$. Using Corollary~\ref{sm-cor} we have: 
\begin{gather*}
\sum_{i=1}^3\MS(V_i)
\leq 
\left(4 \sum_{i=1}^3\delta_i^2+\sum_{i=1}^3\delta_i\right)n
<
\left(4  \left(\sum_{i=1}^3\delta_i\right)^2+\frac{1}{4}\right)n
\\
<
\left(4  \left(\frac{1}{4}\right)^2+\frac{1}{4}\right)n
=
\frac{n}{2}
<
(4\dlow-1)\,n
\end{gather*}
{\bf Case (ii): $\cS$ has one giant component} 
Let $V_1$ be the giant component. Since $1-t_1<1-\dlow<\nicefrac{3}{4}$,
there are at most two other clusters with strictly more than $\nicefrac{n}{4}$ nodes.

\vspace*{0.1in}
\noindent
{\bf Subcase (ii-a): there is one other cluster with strictly more than $\pmb{\nicefrac{n}{4}}$ nodes}
Let this cluster be $V_2$. By Corollary~\ref{sm-cor}, $\sum_{j=3}^{m+1}\MS(V_j)\leq 0$.
Note that $t_2\leq\dlow$. Now, by reusing the calculations of Lemma~\ref{less-than-half}
and using Lemma~\ref{complement} we get
\begin{gather*}
\MS(\cS)=\MS(V_1)+\MS(V_2)+\sum_{j=3}^{m+1}\MS(V_j)
\leq \MS(V_1)+\MS(V_2)
\\
\leq 
\!\!\!\!\!\!\!\!\!\!\!\!\!\!\!\!
\underbrace{\left(2\dlow-\frac{1}{2}\right)\,n}_{\substack{\mbox{by Lemma~\ref{complement} if $t_1>\nicefrac{1}{2}$}\\ \mbox{by Lemma~\ref{less-than-half} if $t_1\leq \nicefrac{1}{2}$} }} 
+ 
\underbrace{\left(2\dlow-\frac{1}{2}\right)\,n}_{\mbox{by Lemma~\ref{less-than-half} since $t_2\leq\dlow$}} 
\,\,\,
\!\!\!\!\!\!\!\!\!\!
=
\,\,\,
\big(4\dlow-1\big)\,n
\end{gather*}
{\bf Subcase (ii-b): there are two other clusters with strictly more than $\pmb{\nicefrac{n}{4}}$ nodes}
Let these clusters be $V_2$ and $V_3$. Then, $\dlow\, n<|V_1|<\nicefrac{n}{2}$. 
By Corollary~\ref{sm-cor}, $\sum_{j=4}^{\,m+1}\MS(V_j)\leq 0$.
Let $t_2=\frac{1}{4}+\delta_2$ and $t_3=\frac{1}{4}+\delta_3$ with $0<\delta_2\leq\delta_3<\frac{1}{2}-\dlow<\nicefrac{2}{100}$.
Thus,
\begin{gather*}
\MS(\cS)\leq \MS(V_1)+\MS(V_2)+\MS(V_3)
\\
\leq
\!\!\!\!\!\!\!\!\!\!\!\!\!\!\!\!
\underbrace{\left(2\dlow-\frac{1}{2}\right)\,n}_{\mbox{by Lemma~\ref{less-than-half} since $t_1<\nicefrac{1}{2}$}} 
+
\,\,\,
\underbrace{\left(4\delta_2^2+\delta_2\right)\,n}_{\mbox{by Corollary~\ref{sm-cor}}}
+
\,\,\,
\underbrace{\left(4\delta_3^2+\delta_3\right)\,n}_{\mbox{by Corollary~\ref{sm-cor}}}
\end{gather*}
Since 
$4\delta_2^2+\delta_2+4\delta_3^2+\delta_3<8{\left(\nicefrac{2}{100}\right)}^2+2\left(\nicefrac{2}{100}\right)<2\dlow-\frac{1}{2}$,
we have $\MS(\cS)\leq(4\dlow-1)\,n$.

\vspace*{0.1in}
\noindent
{\bf Case (iii): $\cS$ has two giant components} 
Let $V_1$ and $V_2$ be the two giant components with $t_1=\dlow+\mu_1$ and $t_2=\dlow+\mu_2$ for some 
$0<\mu_1\leq\mu_2<1-2 \dlow$. 
Since $\left|\cup_{j=3}^{m+1} V_i\right|=(1-t_1-t_2)\,n\leq(1-2\dlow)\,n<\nicefrac{n}{4}$, by Corollary~\ref{sm-cor}
$\sum_{j=3}^{m+1}\MS(V_j)\leq 0$.
Now, by reusing the calculations in the proof of the case of $t>\dlow$ of Lemma~\ref{less-than-half}
and using 
Lemma~\ref{complement}
we get
\begin{gather*}
\MS(\cS)=\MS(V_1)+\MS(V_2)+\sum_{j=3}^{m+1}\MS(V_j)
\\
\leq
\!\!\!\!\!\!\!\!\!\!\!\!
\underbrace{\left(2\dlow-\frac{1}{2}\right)\,n}_{\substack{\mbox{by Lemma~\ref{complement} if $t_1>\nicefrac{1}{2}$}\\ \mbox{by Lemma~\ref{less-than-half} if $t_1\leq \nicefrac{1}{2}$}} } 
+ 
\underbrace{\left(2\dlow-\frac{1}{2}\right)\,n}_{\substack{\mbox{by Lemma~\ref{complement} if $t_2>\nicefrac{1}{2}$}\\ \mbox{by Lemma~\ref{less-than-half} if $t_2\leq \nicefrac{1}{2}$}} } 
\,\,\,
\!\!\!\!\!\!\!\!\!\!\!\!
= 
\,\,\,\,\,\,
(4\dlow-1)\,n
\qedhere
\end{gather*}
\end{proof}

\subsection{Additive Approximations for Locally Dense Graphs}
\label{RL}

Using the algorithmic version of the regularity lemma in~\cite{FK96} we can show that
if the given graph is dense then, for any given constant $\alpha>0$, there is a polynomial-time algorithm that 
returns a solution of modularity value at least $\opt-\alpha$.

\begin{proposition}[constant addiditive error]\label{dense}
Suppose that the given graph $G=(V,E)$ is dense, \IE, 
$m=|E|=\delta n^2$ for some constant $0<\delta<\nicefrac{1}{2}$.
Then, for any given constant $0<\alpha<1$, there is a polynomial-time algorithm that 
returns a solution of value at least $\opt-\alpha$.
\end{proposition}

\begin{proof}
The $\ell$-way cut problem is defined as follows. We are given an weighted graph $G=(V,E)$ with 
$w(u,v)\in\R$ being the weight of the edge $\{u,v\}\in E$. A valid solution is a partition of $V$ to $\ell$ subsets 
$\cS=\left\{S_1,S_2,\dots,S_\ell\right\}$, and the goal is to {\em maximize} the sum of weights of those edges whose end-points
are in different subsets, \IE, maximize $w(\cS)=\sum_{\{u,v\}\in E(\cS)} w(u,v)$, 
where $E(\cS)=\left\{\,\{u,v\}\,|\,\forall\,1\leq j\leq\ell\colon \left|\,\{u,v\}\,\cap S_j\right|\neq 2 \right\}$ is the set of all ``inter-partition'' edges. 
The following result was proved in~\cite{FK96}. 

\begin{theorem}{\rm~$\!\!\!$\cite{FK96}}\label{fk}
Given an weighted graph $G=(V,E)$ of $n$ nodes and 
any constant $0<\eps<1$ there is a polynomial-time algorithm $A_\epsilon$ which, 
computes a partition $\cS_\eps$ of $V$ such that 
\[
w(\cS_{\eps})\geq w(\cS^*)- \eps n^2
\]
 where $\cS^*$ is an optimal (maximum weight) partition.
\end{theorem}

Equation~\eqref{eq:hh} can be used to assign edge weights to cast our modularity clustering problem as an $\ell$-way cut 
problem in the following manner. Consider the complete graph on $n$ nodes ($K_n$) and let 
$w_{u,v}= 2\,\delta\left(\frac{d_u d_v}{2m} - a_{u,v}\right)$ for the edge $\{u,v\}$ of $K_n$.
Then, for a partition $\cS=\left\{S_1,S_2,\dots,S_\ell\right\}$ of the nodes of $K_n$, 
\[
w(\cS)=\sum_{\{u,v\}\in E(\cS)} \hspace*{-0.2in} 2\,\delta\left(\frac{d_u d_v}{2m} - a_{u,v}\right)
=2\,m\,\delta\,\M(\cS)=2\,\delta^2 n^2\, \M(\cS)
\]
Let $\apx_\eps$ be the objective value of an approximate solution of the modularity clustering problem on the given graph
obtained by using the $\ell$-way partitioning of Theorem~\ref{fk} with $\eps=2\,\alpha\,\delta^2$. Then, 
\[
2\delta^2 n^2\apx_\eps \geq 2\delta^2 n^2\opt - \eps n^2 \, \equiv \, \apx_\eps \geq \opt-\alpha
\qedhere
\]
\end{proof}

\section{Hardness and Approximation Algorithms for Sparse Graphs}
\label{nphard}

\subsection{$\NP$-hardness}

Brandes \EA~\cite{BDGGHNW07-1} proved $\NP$-hardness of the $2$-clustering problem
provided nodes with very large degrees are allowed in the input graph. 
Thus it is not a priori clear whether calculating modularity on very sparse graphs becomes easy and 
admits an {\em exact} polynomial-time algorithm.
However, we rule out this possibility of exact solution. 
Our construction is similar to that in~\cite{BDGGHNW07-1}, but carefully replaces dense graphs with {\em nicely behaving} sparse graphs. We have to do a more careful analysis of the properties of an optimal $2$-clustering so as to get the following result.

\begin{theorem}\label{nph}
Computing $\opt_2$ is $\NP$-complete even for $d$-regular graphs for any {\em constant} $d\geq 9$.
\end{theorem}

\begin{proof}
The decision version $\DecProb{2BdRegModularity}$ of our problem is as follows:

\begin{quote}
{\em given a $d$-regular graph $G$ and a number $K$, is there a clustering $\cS$ of $G$ into at most two clusters for which $\M(\cS)\geq K$}?
\end{quote}

\noindent
Our reduction is from the minimum graph bisection problem for $4$-regular graphs ($\DecProb{MB4}$): 
{\sl Given a $4$-regular graph $G$ with $n$ nodes (with even $n$) and an integer $c$, is there a 
clustering into two clusters each of $\nicefrac{n}{2}$ nodes such that it ``cuts'' at most $c$ edges, \IE, at most $c$ 
edges have two end-points in different clusters}?
$\DecProb{MB4}$ is known to be $\NP$-complete~\cite{KPH99}. We reduce an instance $G$ of $\DecProb{MB4}$ to an instance of 
$\DecProb{2BdRegModularity}$ in a manner similar to that in~\cite{BDGGHNW07-1}. 
Every node in $G$ is replaced by a copy of an $n$-node $d$-regular graph $H$ such that 
the minimum cut (minimum number of edges in a cut) of $H$ is at least $d$.
Such a family of graphs can be constructed in the following recursive manner:
\begin{itemize}
\item
For $d=2$, the $2$-regular graph, namely a simple cycle consisting of $n$ nodes,
has a minimum cut of $2$ edges.

\item
For $d=3$, consider two simple cycles $H_1=(V_1,E_1)$ and $H_2=(V_2,E_2)$, each consisting of $\nicefrac{n}{2}$ nodes. Consider an arbitrary matching 
between the nodes of $H_1$ and $H_2$ and add the edges corresponding to this matching to obtain a $3$-regular 
graph $H=(V,E)$. Consider an arbitrary subset of nodes $V'\subset V$ of $H$. Then,
\begin{itemize}
\item
If $V'\cap V_1\neq\emptyset$ and $V'\cap V_2\neq\emptyset$, then the number of cut edges is at least $4$.

\item
Otherwise, assume that $V'\cap V_1=\emptyset$ (the other case is symmetric) and thus $\emptyset\subset V'\subseteq V_2$. 
If $V'=V_2$ then the number of cut edges is exactly $\nicefrac{n}{2}>2$. Otherwise, the number of cut edges is at least 
$2$ (corresponding to two edges of the cycle in $H_2$) plus $1$ (corresponding to one of the matching edges added).
\end{itemize}

\item
For $d>3$, a recursive construction of such graphs follows in a similar manner: take such a $(d-2)$-regular
graph $H$ on $n$ nodes for which the inductive hypothesis applies and add a simple cycle 
to $H$ all of whose edges are different from those in $H$. Consider a cut in this graph. By the 
induction hypothesis the cut contains at least $d-2$ edges of $H$ and at least $2$ additional edges 
of the new cycle added to $H$.
\end{itemize}
Let $H_v$ denote the copy of $H$ corresponding to the node $v\in G$. Delete two independent edges (\IE, edges without any common end-points) 
in $H_v$. The four edges connected to $v$ are now connected to the four endpoints of these deleted edges. 
This is done in order to make the final graph $G'$ $d$-regular\footnote{This is one step that is different from the reduction 
in~\cite{BDGGHNW07-1}, where every node in $G$ is 
replaced by a copy of $K_n$ producing the final graph with non-constant degrees. Since $G$ is $4$-regular, we need $d>8$.}.
Note that the number of nodes in the transformed graph $G'$ is $n^2$, whereas the number of edges is $m=\frac{d\,n^2}{2}$. 
Since two edges are removed from $H$ in the construction, the minimum cut in each modified copy of $H$ is at least $d-2$. 
The correctness of the reduction follows by showing that 
$\DecProb{MB4}$ has a solution with at most $c$ cut edges if and only if $\M(\cS^*)\geq \frac{1}{2} -\frac{c}{m}$.

Let $\cS^*$ be an optimal clustering of $G'$.

\begin{lemma}\label{exact2}
$\cS^*$ has exactly two clusters and $\M(\cS^*)>0$.
\end{lemma}

\begin{proof}
It suffices to show a clustering $\cS=\{C_1,C_2\}$ such that $\M(\cS)>0$. 
To this end, let $C_1=\{H_v\}$ for some $v$, and let $C_2$ contain the rest. Then using Equation~\eqref{eq:3} and 
the fact that $d(n-1)>4$, we get 
\[
\M(\cS) = \frac{D_1(2m-D_1)}{2m^2} - \frac{4}{m} = \frac{dn(dn^2-dn)}{\frac{d^2n^4}{2}} - \frac{4}{\frac{dn^2}{2}} = \frac{2d(n-1)-8}{dn^2} > 0
\qedhere
\]
\end{proof}

The next lemma shows how to normalize a solution without decreasing the modularity value.
Part~{\bf (a)} of the lemma states that $\cS^*$ cannot have any copy of $H$ split across clusters,
whereas part~{\bf (b)} implies that any optimal clustering has to be a bisection of the graph.

\begin{lemma}\label{lm:completeH}
It is possible to normalize an optimal solution $\cS^*$ without decreasing the modularity value 
such that the following two conditions hold: 
\begin{description}
\item[(a)]
For every $v\in G$, there exists a cluster $C\in\cS^*$ such that $H_v\subseteq C$.

\item[(b)]
Each cluster in $\cS^*$ contains exactly $\nicefrac{n}{2}$ copies of $H$.
\end{description}
\end{lemma}

\begin{proof}
Suppose the set of nodes of $G'$ is partitioned into three subsets $A$, $B$ and $C$. Let $\cS_1 = \left\{A\cup C, B\right\}$, and 
we want to transfer the nodes in $C$ to the other cluster to form the clustering $\cS_2 = \{A, B\cup C\}$. 
For any two disjoint subsets $X$ and $Y$ of nodes of $G'$, let 
$m_{XY}$ denote the number of edges one of whose endpoints is in $X$ and the other in $Y$ 
and $D_X = \sum_{v\in X} d_v$ denote the sum of degrees of nodes in $X$.
Then, using Equation~\eqref{eq:2} or Equation~\eqref{eq:3}, the gain in modularity $\Delta = \M(\cS_2) - \M(\cS_1)$ can be simplified and written as
$\Delta = \dfrac{(D_A - D_B)D_C}{2m^2} + \dfrac{m_{BC}-m_{AC}}{m}$.
Using the fact that $G'$ is $d$-regular and substituting for $m$, we get
\begin{equation}
\frac{dn^4}{2}\Delta = d\,|C|\,\big(\,|A|-|B|\,\big) + n^2\,\big(m_{BC}-m_{AC}\big)
\label{eq:gain} 
\end{equation}

\noindent
{\bf (a)} 
Let us assume that there exists a $v\in G$ such that $H_v$ is split across clusters in the optimal clustering 
$\cS^*=\{C_1,C_2\}$. Without loss of generality, we can assume that $|C_1\setminus H_v|\geq|C_2\setminus H_v|$. 
We will transfer the part of $H_v$ in $C_1$ from $C_1$ to $C_2$. Let $A = C_1\setminus H_v$, $B = C_2$, $C = H_v\setminus C_2$, and $|C|=k$.
Then the part of $H_v$ in $C_2$ has a size of $n-k$. 
By our assumption, 
\[
|A|-|B|=|C_1\setminus H_v|-|C_2| = |C_1\setminus H_v|-|C_2\setminus H_v|-|H_v\setminus C_2|\geq -(n-k)
\]
Substituting this in Equation~\eqref{eq:gain}, we get
\[
\frac{dn^4}{2}\,\Delta\geq d[-k(n-k)]+n^2(m_{BC}-m_{AC})
\]
Now, since the original graph $G$ was $4$-regular, at most $4$ extra inter-cluster edges will appear after the transfer. 
Thus, $m_{AC}\leq 4$. The term $m_{BC}$ represents the number of edges between $C_2$ and $H_v\setminus C_1$, 
which is at least the number of edges between the two parts of $H_v$. Thus, $m_{BC}$ is at least the number of edges
in a minimum cut of $H_v$ which is at least $d-2$. This gives
\[
\frac{dn^4}{2}\,\Delta\geq -dk(n-k)+n^2(d-2-4)\geq -\,d\,\frac{n^2}{4}+(d-6)n^2=\frac{(3d-24)n^2}{4}>0
\]
where the second inequality is due to the fact that $k(n-k)$ is maximized when $k=n/2$, and the last inequality is satisfied when $d\geq 9$. 
Hence the modularity can be strictly improved by putting each copy of $H$ completely in a cluster.

\vspace*{0.1in}
\noindent
{\bf (b)} 
By the previous part, each $H_v$ is contained completely in one cluster of $\cS^*=\{C_1,C_2\}$. 
Now assume that $C_1$ has more copies of $H$ than $C_2$. Since $n$ is even, this implies that $C_1$ has 
at least two more copies of $H$ than $C_2$. We will create a new clustering by transferring a copy of $H$ from 
$C_1$ to $C_2$. Then the gain in modularity after this transfer is given by Equation~\eqref{eq:gain}, where $C$ 
denotes the transferred copy of $H$, $B=C_2$ and $A=C_1\setminus C$. By our assumption, $|A|-|B|\geq |C|$. 
Therefore we can simplify the first term and get 
$\frac{dn^4}{2}\,\Delta\geq d\,|C|^2 + n^2(m_{BC}-m_{AC})$.
Also, since the original graph $G$ was $4$-regular, at most $4$ extra inter-cluster edges will appear after the transfer. Simplifying and substituting values,
$\frac{dn^4}{2}\,\Delta\geq dn^2 -4 n^2 > 0$.
Hence, the modularity can be strictly improved by balancing out the copies of $H$ in both clusters.
\end{proof}

Armed with the above lemma, one can now prove the $\NP$-completeness of our problem. 
We will use the above construction to reduce an instance $\langle G,c\rangle$ of $\DecProb{MB4}$ to an instance $\langle G',K\rangle$ of 
$\DecProb{2BdRegModularity}$ with $K=\frac{1}{2}-\frac{c}{m}$. Now suppose $\cS^*=\{C_1,C_2\}$ is an optimal 2-clustering of $G'$. Then, 
$\M(\cS^*) = \frac{D_1D_2}{2m^2}-\frac{m_{12}}{m}$.
By Lemma~\ref{lm:completeH}{\bf (b)}, $D_1=D_2=m$. Also, because of Lemma~\ref{lm:completeH}{\bf (a)}, $m_{12}$ only has edges from $G$, thus 
representing a bisection of $G$. Therefore, $m_{12}\leq c$ if and only if
$\M(\cS^*) \geq \frac{1}{2}-\frac{c}{m}=K$.
\end{proof}

\subsection{Large Integrality Gap for an $\ILP$ Formulation}
\label{lp-relax}

\begin{figure}[ht]
\begin{center}
\begin{tabular}{l}
\toprule
{\em maximize} $\displaystyle\frac{\sum_{\{u,v\,:\,u\neq v\}}\left(a_{u,v}-\frac{d_u d_v}{2m}\right)\left(1-x_{u,v}\right)}{2m}\,-\,{\sum_{v\in V}\frac{d_v^{\,2}}{2\,m}}$ \\
{\em subject to}  
$\forall\, u\neq v\neq z\colon x_{u,z}\leq x_{u,v}+x_{v,z}$ \\
\hspace*{0.75in} $\forall\, u\neq v\colon 0\leq x_{u,v}\leq 1$ \\ 
\bottomrule
\end{tabular}
\end{center}
\caption{\label{fff1}$\LP$-relaxation of modularity clustering~\cite{AK08,BDGGHNW07-1,cor2}.}
\end{figure}

There is an integer linear programming ($\ILP$)
formulation of modularity clustering with arbitrarily many clusters
as shown in Fig.~\ref{fff1}: 
$x_{u,v}=0$ 
if $u$ and $v$ belong to the {\em same} cluster and $1$ otherwise, and 
the ``triangle inequality'' constraints 
$x_{u,z}\leq x_{u,v}+x_{v,z}$
ensure that if $\{u,v\}$ and $\{v,z\}$ belong to the same cluster then 
$\{u,z\}$ also belongs to the same cluster.
Agarwal and Kempe~\cite{AK08} used such an $\LP$-relaxation with several rounding 
schemes for empirical evaluations. However, as we show below, the worst case integrality gap of the
$\LP$-relaxation is at least about the square root of the degree of the graph, thereby ruling out logarithmic 
approximations via rounding such $\LP$-relaxations.

\begin{lemma}\label{integrality}
For every $d>3$ and for all sufficiently large $n$, there exists
a $d$-regular graph with $n$ nodes such that 
the integrality gap of the $\LP$-relaxation in Fig.~\ref{fff1} is $\Omega(\sqrt{d}\,)$. 
\end{lemma}

\begin{proof}
Let $\opt_f$ be the optimal objective value of the $\LP$-relaxation.
For {\em any} graph $G=(V,E)$, a valid fractional solution of the $\LP$-relaxation is
as follows: set $x_{u,v}=\frac{1}{2}$ for {\em every} $\{u,v\}\in E$ and set $x_{u,v}=1$ otherwise.
The value of this fractional solution is precisely $\frac{1}{2}\,-\sum_{v\in V} \frac{d_v^{\,2}}{2m}$.
Thus, in particular, if $G$ is a $d$-regular graph then $\opt_f\geq \frac{1}{2}-\frac{1}{n}$.  

On the other hand, suppose that $G$ is a random $d$-regular graph and 
let $\lambda$ be the second largest eigenvalue of the adjacency matrix $A$ of $G$.
It is well-known that $\lambda<\beta\sqrt{d}$ for some positive constant $\beta$~\cite{FKS89}.
Consider an optimal solution $\emptyset\subset V'\subset V$ of $2$-clustering of $G$
with $0<|V'|=\alpha n\leq n/2$ and  
let $\mbox{cut}(V')$ denote the number of edges between $V'$ and $V\setminus V'$.
By the {\em expander mixing lemma}, we have 
\begin{gather*}
\left|\, \mathrm{cut}(V') - d\frac{(\alpha n)\times (1-\alpha)n}{n}\, \right|
\leq
\lambda \sqrt{(\alpha n) (1-\alpha)n}
\,\,\,
\\
\equiv
\,\,\,
\left|\, \mathrm{cut}(V') - \alpha(1-\alpha)\,d\,n\, \right|
\leq
\lambda \sqrt{\alpha (1-\alpha)}\,n
\end{gather*}
which implies 
$\mbox{cut}(V')\geq\alpha(1-\alpha)\,d\,n -\lambda \sqrt{\alpha(1-\alpha)}\,n>\alpha(1-\alpha)\,d\,n -\beta \sqrt{d}\,n$.
Let $\mbox{uncut}(V')$ denote the number of edges between pairs of nodes in $V'$.
Then, 
$\mbox{uncut}(V')=\frac{\alpha\,d\,n -\mathrm{cut}(V')}{2}<\frac{\alpha^2\,d\,n +\beta\sqrt{d}\,n}{2}$.
Using this in Equation~\eqref{eq:2} (with $m=dn/2$) 
together with Lemma~\ref{ktotwo}~and~\ref{subgraph-selection} shows
\begin{gather*}
\M(V')=\frac{2\times\mbox{uncut}(V')}{d\,n}-\left(\frac{\alpha\,d\,n}{d\,n}\right)^2<\frac{\beta}{\sqrt{d}}
\\
\Longrightarrow\,\,
\opt\leq 2\,\opt_2=4\,\M(V')<\frac{4\,\beta}{\sqrt{d}}
\,\,\Longrightarrow\,\,
\frac{\opt_f}{\opt}=\Omega(\sqrt{d}\,)
\qedhere
\end{gather*}
\end{proof}

\subsection{Logarithmic Approximation}
\label{log-regular}
\label{weighted}

Newman~\cite{N04} extended the modularity measure to weighted graphs in the following manner.
Let $G=(V,E,\ell)$ be the input weighted graph with $\ell:E\mapsto\R^+$ being the function mapping edges to 
{\em non-negative} real-valued weights. Now, if we redefine $d_u=\sum_{\{u,v\}\in E}\ell(u,v)$ as the ``weighted'' degree of the node $u$,  
$m=\sum_{u\in V}d_u$, and $A=\left[a_{u,v}\right]$ as the {\em weighted} adjacency matrix of $G$ 
(\IE, $a_{u,v}=\ell(u,v)$ if $\{u,v\}\in E$ and $0$ otherwise), then Equation~\eqref{eq:11} applies to the weighted case also.
The corresponding modification in Equation~\eqref{eq:2} can be obtained by redefining
$m_i$ as the {\em total weight} of edges whose both endpoints are in the cluster $C_i$, 
$m_{ij}$ as the {\em total weight} of edges one of whose endpoints is in $C_i$ and the other in $C_j$ 
and $D_i = \sum_{v\in C_i} d_v$ as the sum of {\em weighted degrees} of nodes in cluster $C_i$.
It is straightforward to see that Lemma~\ref{ktotwo} holds even for weighted graphs.

{\em We denote the weighted degree, the maximum weighted degree and the average weighted degree of a node $v$ by $d_v$, 
$d_{\max}=\max_{v\in V}\{d_v\}$ and $\Delta=\frac{\sum_{v\in V}d_v}{n}$, respectively}, and, 
for convenience, {\em we normalize\footnote{It is easy to see that the modularity value of any clustering remains unchanged if all weights
are scaled by the same factor.} all the weights such that $\sum_{v\in V}d_v$ is twice the number of edges of $G$}.

\begin{theorem}\label{log-approx}~\\
\noindent
{\bf (a)}
There exists a polynomial time $O(\log d)$-approximation for $d$-regular graphs with $d<\frac{n}{2\ln n}$.

\noindent
{\bf (b)}
There exists a polynomial time $O(\log d_{\max})$-approximation for weighted graphs
$d_{\max}<\frac{\sqrt[5]{n}}{16\ln n}$.
\end{theorem}

\begin{proof}
We begin with the approximation algorithm for regular graphs, which is somewhat easier 
to analyze, and later generalize the results for weighted graphs.
A common theme for both the proofs is the following approach.
By Lemma~\ref{ktotwo} $\opt_2\geq\nicefrac{\opt}{2}$, and thus it suffices to provide a logarithmic approximation
for the $2$-clustering problem on $G$. 
For notational convenience let 
$w_{u,v}=\dfrac{a_{u,v}-\frac{d_u d_v}{2m}}{2\,m}$.
As observed in~\cite{M06-2}, 
letting $x_u\in\{-1,1\}$ be the indicator variable denoting the partition that node $u\in V$ belongs to, 
Equation~\eqref{eq:1} can be rewritten for a $2$-clustering as
$
\M(\cS)=
\sum_{u,v\in V} w_{u,v}\left(1+x_u x_v\right)
=
\sum_{u,v\in V} w_{u,v} x_u x_v
=
\x^{\mathrm{T}}W\x
$
where $\x\in \{-1,1\}^n$ is a column vector of the indicator variables
and $W=\left[w_{u,v}\right]\in\R^{n\times n}$ is the corresponding symmetric matrix. 
The following result is known on quadratic forms.

\begin{theorem}{\rm\cite{CW04}}\label{SDP}
Consider maximizing
$\x^{\mathrm{T}}Z\x$ subject to $\x\in\{-1,1\}^n$, 
where $Z=[z_{i,j}]$ is a $n\times n$ real matrix with $z_{i,i}\geq 0$. 
Then, for any $T>1$, there exists a randomized approximation algorithm whose objective value $\kappa$ satisfies 
$\displaystyle\expect[\kappa]\geq \frac{\max_{\x\in\{-1,1\}^n} \x^{\mathrm{T}}Z\x}{T^2} - 8\,\e^{-\nicefrac{T^2}{2}}\left(\sum_{i\neq j}|z_{i,j}|\right)$.
\end{theorem}

The above approximation does not directly apply to the quadratic form for modularity clustering 
since the diagonal entries are negative for our case. 
Moreover, the lower bound on the optimal value of the quadratic form as used in~\cite{CW04}
depends on $n$ which we would like to avoid.

\vspace*{0.1in}
\noindent
{\bf (a) The Case When the Input Graph is Regular}.

The proof of the following lemma uses a result in 
~\cite{HY07} on the size of a maximum-cardinality matching of a regular graph.
The above lemma is tight in the sense that there exist $d$-regular graphs for which 
$\opt=O\left(\nicefrac{1}{\sqrt{d}}\right)$ (the proof of Lemma~\ref{integrality} shows that 
$d$-regular expanders are one such class of graphs).

\begin{lemma}[Lower Bound for $\opt$]\label{lbq}
If $n>40d^{\,9}$ then $\opt>\frac{0.26}{\sqrt{d}}$, else $\opt>\frac{0.86}{d}-\frac{4}{n}$.
\end{lemma}

\begin{proof}
Consider a maximum-cardinality matching $\{u_1,v_1\},\dots,\{u_k,v_k\}$ of
$G$ of size $k$. 
It is known~\cite{HY07} that for any $d>2$, 
\[
k\geq\left\{
\begin{array}{ll}
\min\left\{\dfrac{n(d^2+4)}{2d^2+2d+4}\,,\,\dfrac{n-1}{2}\right\}, & \mbox{if $d$ is odd} \\
\dfrac{(d^3-d^2-2)n-2d+2}{2(d^3-3d)}\,, & \mbox{otherwise} \\
\end{array}
\right.
\]
which gives $k>0.43\,n$ for any $d$.
We create $k$ clusters $\{V_1,V_2,\dots,V_k\}$ where $V_i=\{u_i,v_i\}$ and 
for each remaining node $u\in V\setminus\left(\cup_{\,i=1}^{\,k} V_i\right)$ we create a cluster $\{u\}$ of one node.
Using Equation~\eqref{eq:2}, we have 
\[
M(\cS)=\sum_{C_i}\,\left[\frac{m_i}{m}-\left(\frac{D_i}{2m}\right)^2\right]
=\sum_{i=1}^k \left(\frac{2}{dn}-\frac{4}{n^2}\right)
-\sum_{i=k+1}^n \frac{1}{n^2}>\frac{0.86}{d}-\frac{4}{n}
\]
For fixed $d$ and $n>40d^{\,9}$, it was shown in~\cite{A97} that
every $d$-regular graph with $n$ nodes has a bisection width
of at most $\left(\frac{d}{2}-0.13\sqrt{d}\right) \left(\frac{n}{2}\right)$.
Consider the partition $\cS$ of $G$ into two clusters $C_1$ and $C_2$ corresponding to such a bisection
with exactly $\nicefrac{n}{2}$ nodes in each cluster.
Then, $m=\frac{dn}{2}$, $D_1=D_2=m$, $m_1,m_2>\left(\frac{d}{2}+0.13\times\sqrt{d}\right) \left(\frac{n}{4}\right)$
and using Equation~\eqref{eq:2} we get
$\M(C_1)=\M(C_2)>\frac{0.13}{\sqrt{d}}$. Consequently, by Lemma~\ref{subgraph-selection}
$\M(\cS)>\frac{0.26}{\sqrt{d}}$.
\end{proof}

We now define the following quantities:
\begin{itemize}
\item
$\D=\sum_{v\in V}|w_{v,v}|$.

\item
$W'=\left[w_{u,v}'\right] \mbox{where } w_{u,v}'=\left\{\begin{array}{ll} 0, & \mbox{if $u=v$} \\ w_{u,v}, & \mbox{otherwise} \\ \end{array}\right.$.

\item
$\W'=\sum_{u,v\in V}|w_{u,v}'|$.
\end{itemize}
Thus, if $\opt_2=\max_{\x\in\{-1,1\}^n} \x^{\mathrm{T}}W\x$ and 
$\opt_2'=\max_{\x\in\{-1,1\}^n} \x^{\mathrm{T}}W' \x$ then $\opt_2'=\opt_2-\D$. 

\begin{lemma}\label{sum-coeff}
$\W'<2$. 
\end{lemma}

\begin{proof}
\begin{gather*}
\W' 
< \!\!\! 
\sum_{u,v\in V}\!\!\big|w_{u,v}\big| 
= \!\!\! 
\sum_{w_{u,v}\geq 0} \!\!\!\! w_{u,v} - \!\!\!\!\sum_{w_{u,v}<0} \!\!\!\! w_{u,v}  
= \hspace*{-0.85in}
\underbrace{2 \left( \!\!\!\!\!\!\! \sum_{\,\,\,\,\,\,\,\,\,\,w_{u,v}\geq 0} \!\!\!\!\!\!\! w_{u,v}\!\!\right)}_{\mbox{since $\displaystyle \!\!\!\sum_{u,v\in V}\!\!\!w_{u,v}\!=\!\!\!\!\sum_{w_{u,v}\geq 0}\!\!\!w_{u,v}\!-\!\!\!\!\!\sum_{w_{u,v}<0}\!\!\!w_{u,v}=0$}} 
\hspace*{-0.9in} < 
\frac{\displaystyle \!\!\!\!\!\!\!\!\!  \sum_{\,\,\,\,\,\,\,\,\,\,\,\,\{u,v\}\in E} \!\!\!\!\!\!\!\!\! a_{u,v}}{m} =2  
\end{gather*}
\end{proof}

Next, we bound $\D$ by observing that, for any $d$,  
$\D=\frac{d^2n}{4m^2}=\frac{1}{n}$.
To complete the proof,  we use the algorithm in Theorem~\ref{SDP} with $Z=W'$. Using Lemmas~\ref{ktotwo},~\ref{lbq}~and~\ref{sum-coeff} 
we get the desired approximation guarantees of Theorem~\ref{log-approx} by choosing $T=\sqrt{4\ln d}$ in 
the algorithm in Theorem~\ref{SDP}. Then we have the following chain of implications for all sufficiently large $d$ and $n$:
\begin{itemize}
\item
$\opt_2'=\opt_2-\D\geq\frac{\opt}{2}-\D>\frac{0.43}{d}-\frac{1}{n}>\frac{0.43}{d}-\frac{1}{2d\ln n}>\frac{0.4}{d}$.

\item
Thus, 
$\dfrac{\W'}{\opt_2'}<\frac{2d}{0.4}=5d$.

\item
Thus, 
$\expect[\kappa]>\frac{\opt_2'}{T^2}-4\,\e^{\!-\frac{T^2}{2}}d\,\opt_2'=\frac{\opt_2'}{4\ln d}-\frac{4d}{d^{2}}\,\opt_2'>\frac{\opt_2'}{4.1\ln d}$.
\end{itemize}
Thus, the final modularity value achieved is at least 
\begin{gather*}
\frac{\opt_2'}{4.1\ln d}-\D
=
\frac{\opt_2-\D}{4.1\ln d}-\D
\\
=
\frac{\opt_2}{4.1\ln d}\,-\,\left(1+\frac{1}{4.1\ln d}\right)\left(\frac{0.4}{d}+\frac{1}{n}\right)\left(\frac{d}{d+0.4n}\right)
\\
>
\left( \frac{1}{4.1\ln d}\,-\, \left(1+\frac{1}{4.1\ln d}\right)\left(\frac{1}{1+0.4\frac{n}{d}}\right) \right) \opt_2
\\
>
\left( \frac{1}{4.1\ln d}\,-\, \left(1+\frac{1}{4.1\ln d}\right)\left(\frac{1}{1+0.8\ln n}\right) \right) \opt_2
>
\frac{\opt_2}{4.2\ln d}
>
\frac{\opt}{8.4\ln d}
\end{gather*}

\vspace*{0.1in}
\noindent
{\bf (b) The Case When the Input Graph is Weighted}

Since the given graph can be assumed to be connected, $\Delta\geq 1-\frac{1}{n}$. 
We want to design an $O\left(\log d_{\max}\right)$-approximation algorithm assuming 
$d_{\max}<\frac{\sqrt[5]{n}}{16\ln n}$. Again, we first provide a lower bound for $\opt$.

\begin{center}
\begin{figure}[ht]
\begin{center}
\begin{tabular}{l}\toprule
(* $\cS$ denotes the set of clusters *) \\
(* initialization *) \\
$\cS=\emptyset\,$; $V''=V\,$; $E''=E'=\big\{\,\{u,v\}\,|\,\{u,v\}\in E \, \& \, \ell(u,v)<\nicefrac{1}{2} \big\} \,$; $\forall\, u\in V\colon C_u=\emptyset$ \\
(* Algorithm *) \\
{\bf while} the graph $(V'',E'')$ contains at least one edge {\bf do} \\
\hspace*{0.3in} pick a node $v\in V''$ that maximizes $L(v)=\sum_{\{u,v\}\in E''}\ell(u,v)$ \\
\hspace*{0.3in} $C_v=\{v\}\cup \{u\,|\,\{u,v\}\in E''\}\,$; add the new cluster $C_v$ to $\cS$ \\
\hspace*{0.3in} $V''=V''\setminus C_v\,$; $E''=(V''\times V'')\cap E'$ \\
{\bf endwhile} \\
{\bf for} every $v\in V''$ {\bf do} \\
\hspace*{0.3in} add the cluster $\{v\}$ to $\cS$ \\ 
{\bf endfor} \\
\bottomrule
\end{tabular}
\end{center}
\caption{\label{f1}Greedy algorithm for computing lower bounds for weighted graphs.}
\end{figure}
\end{center}

\begin{lemma}[Lower bound on $\opt$ for weighted graphs]\label{sk}
If $d_{\max}<\dfrac{\sqrt[5]{n}}{16\ln n}$ then $\opt>\dfrac{1}{8\,d_{\max}}$.
\end{lemma}

\begin{proof}
We execute the greedy algorithm on $G'$ as shown in Fig.~\ref{f1}.
Note that the graph $G'=(V,E')$ has a maximum weighted degree of precisely $d_{\max}$, 
The number of nodes adjacent to any node $v$ in $G'$ is at most $2\,d_v\leq 2\,d_{\max}$, 
and $\ell(E')=\sum_{\{u,v\}\in E'}\ell(u,v)=m-\sum_{\{u,v\}\in E\setminus E'}\ell(u,v)\geq \nicefrac{m}{2}$.

Let $L(C_v)=\sum_{\substack{u,v\in C_v \\ u\neq v}}\ell(u,v)$. 
Since the weight of any edge in $E'$ is at least $\nicefrac{1}{2}$, 
it is easy to see that during each selection of cluster $C_v$, $L(C_v)$ is at least 
$\nicefrac{1}{d_{\max}}$ times the total weight of edges whose one end-point was in $C_v$. 
Thus, $\sum_{C_v}\! L(C_v)\geq \dfrac{\ell(E')}{d_{\max}+1}\geq\dfrac{m}{2\,\left(d_{\max}+1\right)}=\dfrac{n\,\Delta}{2\,\left(d_{\max}+1\right)}$.
Note that for all sufficiently large $n$,  
\[
w_{u,v}
=
\left\{
\begin{array}{ll}
\dfrac{\ell_{u,v}-\frac{ d_u d_v}{n\, \Delta}}{n\, \Delta}
\geq
\dfrac{\ell_{u,v}-\frac{\left(d_{\max}\right)^2}{n\, \Delta}}{n\, \Delta}
\geq 
\dfrac{\ell_{u,v}}{2\, n\, \Delta}\,\,,
& \mbox{if $\{u,v\}\in E$} \\
& \\
\dfrac{-\,d_u d_v}{\left(n\, \Delta\right)^2} 
\geq
\dfrac{-\left( d_{\max}\right)^2}{n^2\Delta^2}
\geq
-\,\dfrac{1}{256\, n^{1.6} \ln^2 n\, \Delta^2}\,\,,
& \mbox{otherwise.} \\
\end{array}
\right.
\]
Thus, for all sufficiently large $n$, we have 
\begin{gather*}
\M(\cS)
=
\sum_{v}\M(C_v)\,- \hspace*{-0.3in} \sum_{u\in V\,\setminus \left(\bigcup_{\,v} C_v\right)} \hspace*{-0.3in} w_{u,u} 
\geq 
\dfrac{\displaystyle\sum_{\substack{C_v\in\cS\\C_v\neq\emptyset}}
\!\! \left(\!\!\!\!\!\! \sum_{\,\,\,\,\,\,\,\,\substack{u,v\in C_v \\ \{u,v\}\in E}} \!\!\!\!\!\! \ell_{u,v}\right)}
{2\, n\, \Delta}
\,-\,
\dfrac{ n\left(d_{\max}\right)^2 }{256\, n^{1.6}\, \ln^2 n\, \Delta^2} 
\end{gather*}
\begin{gather*}
\geq
\dfrac{\sum_v L(C_v)}{2 n \Delta}
\,-\,
\dfrac{ n\left(d_{\max}\right)^2 }{512\, n^{1.6}\, \ln^2 n\, \Delta^2} 
\,-\,
\dfrac{ n\left(d_{\max}\right)^2 }{256\, n^{1.6}\, \ln^2 n\, \Delta^2} 
\geq
\dfrac{\frac{n \Delta}{2 \left(d_{\max}+1\right)}}{2 n \Delta}
\,-\,
\dfrac{1}{512\, n^{\nicefrac{1}{5}}\,\ln^4 n}
\\
=
\dfrac{1}{4 \left(d_{\max}+1\right)}
\,-\,
\dfrac{1}{512\,n^{\nicefrac{1}{5}}\,\ln^4 n}
>\dfrac{1}{8\,d_{\max}}
\qedhere
\end{gather*}
\end{proof}

\noindent
Since $d_{\max}<\frac{\sqrt[5]{n}}{16\ln n}$ and $\Delta\geq 1-\frac{1}{n}$,
$\D\leq 
\frac{n\left(d_{\max}\right)^2}{2\left(n\,\Delta\right)^2}
=
\frac{1}{2n} \left(\frac{d_{\max}}{\Delta}\right)^2
\leq\frac{1}{512\,{n}^{\nicefrac{3}{5}}\,\ln^2 n}$.
Selecting $T=\sqrt{16\ln d_{\max}}$ in Theorem~\ref{SDP}, we have the following chain of implications:
\begin{itemize}
\item
$\opt_2'
=
\opt_2-\D
\geq
\dfrac{\opt}{2}-\D
=
\dfrac{1}{16\,d_{\max}} \, - \, \dfrac{1}{512\,{n}^{\nicefrac{3}{5}}\,\ln^2 n}
>
\dfrac{1}{17\,d_{\max}}$.

\item
Thus,
$\dfrac{\W'}{\opt_2'}<34\,d_{\max}$.

\item
Thus,
$\expect[\kappa]>\dfrac{\opt_2'}{T^2}-34\,\e^{-\frac{T^2}{2}}\,d_{\max}\,\opt_2'>\dfrac{\opt_2'}{17\,\ln d_{\max}}$.
\end{itemize}
and thus the final modularity value achieved is at least 
\[
\frac{\opt_2}{17\ln d_{\max}}-\D=\frac{\opt}{O(\ln d_{\max})}
\qedhere
\]
\end{proof}

\section{Other Results}

\subsection{Modularity Clustering for Directed Weighted Graphs}
\label{digraph}

Leicht and Newman~\cite{LN08} generalized the modularity measure to weighted directed graphs
in the following manner. Let $G=(V,E,\ell)$ be the input {\em directed} graph 
with $\ell\colon E\mapsto\R^+$ being the function mapping edges to non-negative weights.
For a node $v\in V$, let $d^{\rm in}_v$ and $d^{\rm out}_v$ denote the {\em weighted in-degree} 
and the {\em weighted out-degree} of $v$, respectively. 
Let $m=\sum_{v\in V}d^{\rm in}_v+\sum_{v\in V}d^{\rm out}_v$ and 
let $A=[a_{u,v}]$ denote the weighted adjacency matrix of $G$,
\IE, $a_{u,v}=\ell(u,v)$ if $(u,v)\in E$ and $a_{u,v}=0$ otherwise.
Note that the matrix $A$ is {\em not} necessarily symmetric now. 
Then, Equation~\eqref{eq:11} computing the modularity value of a cluster $C\subseteq V$ 
needs to be modified as
\[
\M(C)=\frac{1}{m}\left(\sum_{\,u,v\in C}\left(a_{u,v}-\frac{d^{\rm out}_u\, d^{\rm in}_v}{m}\right)\right)
\]
With some effort, we show that we can extend all our complexity results 
for undirected networks to directed networks.
Let 
$\Delta=\dfrac{\sum_{v\in V} d^{\iin}_v}{n}=\dfrac{\sum_{v\in V} d^{\out}_v}{n}$
denote the average weighted degree of nodes of $G$, and let 
$\displaystyle d^{\iin}_{\max}=\max_{v\in V} d^{\iin}_v$ and 
$\displaystyle d^{\out}_{\max}=\max_{v\in V} d^{\out}_v$ denote the maximum weighted in-degree and maximum weighted out-degree, respectively,
of nodes in $G$.
For convenience, we normalize all the weights such that 
$\sum_{v\in V}d^{\iin}_v+\sum_{v\in V}d^{\out}_v$ is {\em exactly} twice the number of directed edges of $G$.
Since the given graph can be assumed to be weakly-connected,
$\Delta\geq 1-\frac{1}{n}$. 

\begin{theorem}$\!\!\!\!$\footnote{We made no serious attempts to optimize various constants in this theorem.}\label{ddir}~\\
\noindent
{\bf (a)}
Computing $\opt_2$ is $\NP$-complete even if every node $v$ has 
$d^{\iin}_v=d^{\out}_v=d$,
for any fixed $d\geq 9$.

\noindent
{\bf (b)}
It is $\NP$-hard to approximate the $k$-clustering problem, for any $k$,  
within a factor of $1+\eps$ for some constant $\eps>0$
even if every node of the given directed graph has 
$d^{\iin}_v=d^{\out}_v=n-4$.

\noindent
{\bf (c)}
There is an $O(\log d)$ approximation algorithm for {\em unweighted} directed graphs 
if the in-degree and out-degree of all nodes 
is exactly the same, say $d$, and $d\leq \frac{n}{100\ln n}$.

\noindent
{\bf (d)}
There is an $O\left(\log \left(d^{\iin}_{\max}+d^{\out}_{\max}\right)\,\right)$-approximation algorithm for weighted graphs
provided $\max\big\{\,d^{\iin}_{\max},d^{\out}_{\max}\,\big\} \leq \frac{\sqrt[5]{n}}{64\ln n}$.
\end{theorem}

\begin{proof}
Remember that 
\begin{equation}
\M(C)=
\frac{1}{m} \left(
\sum_{\,u,\,v\,\in\, C} 
\left(
a_{u,v}-\frac{d^{\rm out}_u d^{\rm in}_v}{m}
\right)
\right) 
\label{eq:11dir}
\end{equation}
The corresponding modification in Equation~\eqref{eq:2} is
\begin{equation}
\M(\cS)=
\sum_{C_i\in\,\cS}\, 
\left(
\frac{m_i}{m} - \left( \frac{D_i^{\rm in}\times D_i^{\rm out}}{m^2} \right) 
\right) 
\label{eq:2dir}
\end{equation}
where 
$D_i^{\rm in}=\sum_{v\in C_i} d^{\rm in}_v$,
$D_i^{\rm out}=\sum_{v\in C_i} d^{\rm out}_v$
and 
$m_i$ as the total weight of edges whose both endpoints are in the cluster $C_i$.
Finally, since 
$
\sum_{v\in V} \left(a_{u,v}-\frac{d^{\rm out}_u d^{\rm in}_v}{m}\right)
\linebreak
=
\sum_{v\in V} \left(a_{u,v}-\frac{d^{\rm in}_u d^{\rm out}_v}{m}\right)
=0$ 
for any $u\in V$,
we can alternatively express $\M(C)$ as 
$\displaystyle \M(C)=\frac{1}{m} \left( \sum_{u\in C,\,v\not\in C} \left( \frac{d^{\rm out}_u d^{\rm in}_v}{m}-a_{u,v} \right)\right)$.
Thus, Equation~\eqref{eq:3} now becomes
\begin{equation}
\M(\cS)= \sum_{C_i,C_j}{\left(\frac{D^{\rm out}_i D^{\rm in}_j}{m^2} - \frac{m_{ij}}{m}\right)}
\label{eq:3dir}
\end{equation}
where $m_{ij}$ as the total weight of the edges directed from $C_i$ to $C_j$. 

\vspace*{0.1in}
\noindent
{\bf (a) \& (b)} 
These two results follow by the following easy observation.
Consider a given undirected unweighted graph $G$ with $n$ nodes and $m$ edges, and 
let $\widetilde{G}$ be the directed graph obtained by replacing each edge $\{u,v\}$ of $G$ by two 
directed edges $(u,v)$ and $(v,u)$, each of weight $1$;
thus $\widetilde{m}=\sum_{v\in V}d^{\iin}_v + \sum_{v\in V} d^{\out}_v=4m$. 
Let $\widetilde{A}=[\widetilde{a}_{u,v}]$ be the adjacency matrix of $\widetilde{G}$,
and 
$\widetilde{d}^{\iin}_v$
and 
$\widetilde{d}^{\out}_v$
be the in-degree and out-degree of the node $v$ in $\widetilde{G}$.
Then, it is easy to see that every clustering of $G$ of modularity value $x$ 
translates to a corresponding clustering of $\widetilde{G}$ of the same modularity value and vice versa.

\vspace{0.1in}
\noindent
{\bf (c) \& (d)}
It is easy to see that the proof of Lemma~\ref{ktotwo} works for directed networks as well
by using Equation~\eqref{eq:3dir} instead of Equation~\eqref{eq:3} in the proof. Thus again it
suffices to approximate $\opt_2$.

Let $W=\left[w_{u,v}\right]\in\R^{n\times n}$ be the matrix whose entries are defined by 
$w_{u,v}=\dfrac{a_{u,v}-\displaystyle\frac{d^{\out}_u d^{\iin}_v}{m}}{2m}$.
Then, 
letting $x_u\in\{-1,1\}$ be the indicator variable denoting in which partition the node $u\in V$
belongs, 
Equation~\eqref{eq:11dir} can be rewritten for a $2$-clustering of directed networks as
\begin{gather*}
\M(\cS)=
\sum_{u,v\in V} \!\!\! w_{u,v}\left(1+x_u x_v\right)
=
\sum_{u,v\in V} \!\!\! w_{u,v} x_u x_v
\\
=
\x^{\mathrm{T}}W\x
=\x^{\mathrm{T}} \left( \frac{W+W^{\mathrm{T}}}{2} \right) \x
=\x^{\mathrm{T}} W' \x
\end{gather*}
where $W'=\frac{W+W^{\mathrm{T}}}{2}=[w'_{u,v}]$ is a {\em symmetric} matrix.
Note that $w_{u,v}'=\dfrac{\delta_{u,v}-\displaystyle\frac{d_u^{\out}d_v^{\iin}+d_u^{\iin}d_v^{\out}}{2m}}{2m}$ 
where $\delta_{u,v}$ is given by:
\[
\delta_{u,v}=\delta_{v,u}=\left\{
\begin{array}{ll}
1, & \mbox{if both $(u,v)\in E$ and $(v,u)\in E$} \\
0, & \mbox{if both $(u,v)\not\in E$ and $(v,u)\not\in E$} \\
\nicefrac{1}{2}, & \mbox{otherwise.} \\
\end{array}
\right.
\]
Let $\widehat{W}=\left[\widehat{w_{u,v}}\right]$ be the real symmetric matrix defined by
$\widehat{w_{u,v}}
=
\left\{
\begin{array}{ll}
0, & \mbox{if $u=v$} \\
w_{u,v}', & \mbox{otherwise.} \\
\end{array}
\right.
$
As in the proof of Theorem~\ref{log-approx}, 
it follows that 
$\sum_{u,v\in V}\widehat{w_{u,v}}<2$. 
For notational convenience, define 
$\D=\mbox{trace}\left(\widehat{W}-W'\right)=\sum_{u\in V}w_{u,u}'$
and 
$\displaystyle\opt_2'=\max_{\x\in\{0,1\}^n} \x^{\mathrm{T}}\widehat{W}\x$. 

\vspace*{0.1in}
\noindent
{\bf (c) $G$ is an unweighted directed graph with $d^{\iin}_v=d^{\out}_v=d$ for every node $v$,
and $d\leq \frac{n}{5\ln n}$}. 

The proof of Theorem~\ref{log-approx} on the quadratic form $\max_{\x\in\{0,1\}^n} \x^{\mathrm{T}}\,\widehat{W}\,\x$ 
gives an approximation factor of $\gamma\ln d$, for some constant $\gamma>0$, for our directed network provided we can show that
\begin{itemize}
\item
$\dfrac{\opt_2'}{\gamma\,\ln d}-\D=\Omega\left(\dfrac{\opt_2'}{\gamma\,\ln d}\right)$, and 

\item
$\opt_2=\Omega\left(d^{-c}\right)$ for some constant $c>0$. 
\end{itemize}
Let $H$ be the undirected graph obtained from the given graph $G$
by ignoring the direction of the edges and removing parallel edges (if any);
every node in $H$ has a degree between $d$ and $2d$. 
Greedily pick a {\em maximal} matching in $H$, each time selecting an edge
and deleting all (at most $4d-1$) edges that have a common end-point with the picked edge.
Such a matching contains at least $\dfrac{(nd)/2}{4d}=\dfrac{n}{8}$ edges, each of weight at least
$\dfrac{1}{4m}-\dfrac{8d^2}{4m^2}=\dfrac{1}{8dn}-\dfrac{1}{2n^2}$ in $G$.
Consider the clustering of $G$ where each edge in the matching is a separate cluster of two nodes, and 
each of the remaining nodes is a separate cluster of one node. 
The modularity value of this solution is {\em at least} 
\[
\left(\frac{1}{8dn}-\frac{1}{2n^2}\right)\frac{n}{8}\,-\,\mbox{trace}\left(W'-\widehat{W}\right)
\geq
\frac{1}{64d}\,-\,\frac{1}{16n}\,-\,\frac{1}{2n}
\]
Thus, $\opt_2'\geq \frac{1}{128\,d}-\frac{9}{32\,n}=\Omega\left(d^{-1}\right)$.
Moreover, since 
$d\leq \frac{n}{100\ln n}$
we have
\[
\frac{\opt_2'}{\ln d}-\D=\frac{\opt_2'}{\ln d}-\frac{1}{2n}
=
\Omega\left(\frac{\opt_2'}{\ln d}\right)
\]
{\bf (d) $\max\big\{d^{\iin}_{\max},d^{\out}_{\max}\big\}< \frac{\sqrt[5]{n}}{64\ln n}$}. 

Let $G''=(V,E'')$ be the undirected weighted graph obtained from $G$ whose adjacency matrix 
is $W''=\left[w_{u,v}''\right]$ with 
$
w_{u,v}''
=\left\{
\begin{array}{ll}
w_{u,v}'-\frac{1}{2}, & \mbox{if $\delta_{u,v}=1$} \\
w_{u,v}', & \mbox{otherwise.} \\
\end{array}
\right.
$
Since $w_{u,v}''\geq w_{u,v}'$, it suffices to show an approximation 
for
$\max_{\x\in\{0,1\}^n} \x^{\mathrm{T}}W'' \x$. 
The algorithm in the proof of Theorem~\ref{lbq}(b) 
with $W=W''$ can now be appropriately modified to obtain the desired approximation 
if one identified the quantity $d_{\max}$ in that proof with $d^{\iin}_{\max}+d^{\out}_{\max}$. 
\end{proof}

\subsection{Alternative Modularity Measure: the $\mami$ Objective}
\label{minob}

Exact or approximate solutions to the modularity measure may produce many {\em trivial} clusters
of single nodes. For example, the following proposition shows that for a large class of
graphs there exists a clustering in which every cluster except one consists of a single node 
gives a modularity value that has a modularity value of {\em at least} $25$\% of the optimal.

\begin{proposition}\label{trivialclust}
There exists a clustering for a graph $G$ in which every cluster except one consists of a single node and 
whose modularity value is at least $25$\% of the optimal if 
\begin{itemize}
\item
$G$ is $d$-regular with $d<\frac{n}{2\ln n}$, or 

\item
$G$ is an undirected weighted graph with $d_{\max}<\frac{\sqrt[5]{n}}{16\,\ln n}$.
\end{itemize}
\end{proposition}

\begin{proof}
Let $\big\{V',V\setminus V'\big\}$ be an optimal $2$-clustering of $G$.
By Lemma~\ref{ktotwo}, $\opt_2\geq\nicefrac{\opt}{2}$. By Lemma~\ref{subgraph-selection} $\M(V')=\nicefrac{\opt_2}{2}=\nicefrac{\opt}{4}$.
Suppose that we replace the cluster $V\setminus V'$ by $|V\setminus V'|$ trivial clusters each of a single node, and let $C$ be this new clustering
If $G$ is $d$-regular, then $\M(C)=\M(V')-\D=\frac{\opt}{4}\,-\,\frac{1}{n}$.
By Lemma~\ref{lbq}, $\opt>\frac{0.86}{d}-\frac{4}{n}$, and thus $\M(C)=\frac{\opt}{4}-\mathrm{o}(1)$. 
Similarly, for the case when $G$ is undirected weighted with $d_{\max}<\frac{\sqrt[5]{n}}{16\,\ln n}$, 
the proof of Theorem~\ref{log-approx} shows that $\D\leq\frac{1}{512\,{n}^{\nicefrac{3}{5}}\,\ln^2 n}$, and thus 
$\M(C)=\M(V')-\D\geq\frac{\opt}{4}\,-\,\frac{1}{512\,{n}^{\nicefrac{3}{5}}\,\ln^2 n}$.
By Lemma~\ref{sk} $\opt>\frac{1}{8\,d_{\max}}$, and thus again $\M(C)=\frac{\opt}{4}-\mathrm{o}(1)$. 
\end{proof}

We investigate one alternative to overcome such a shortcoming: define the modularity of the network as the {\em minimum} of 
the modularities of individual clusters. Equation~\eqref{eq:1} now becomes
\[
\M^{\mami}(\cS)=\min_{C_i\in\cS}\M(C_i)
\]
We will add the superscript ``{\sf max-min}'' to differentiate the relevant quantities for
this objective from the usual summation objective discussed before, \EG, we will
use $\opt^{\mami}$ instead of $\opt$.
In a nutshell, our results in the following lemma show that 
the $\mami$ objective indeed avoids generating trivial clusters (Lemma~\ref{discur}(a)), and 
the optimal objective value for $\mami$ objective 
is precisely scaled by a factor of $2$ from that of the SUM objective, thereby 
keeping the overall quantitative measure the same (Lemma~\ref{nomorethan2}(b)).

\begin{lemma}\label{discur}\label{nomorethan2}
Let $G$ be a weighted undirected graph with $m$ edges and maximum degree $d_{\max}$. Then,
the following claims hold:

\vspace*{0.1in}
\noindent
{\bf (a)} 
No optimal solution for $\mami$ objective has a cluster with
fewer than $\frac{4\,m\,\opt^{\mami}}{d_{\max}}$ nodes.

\vspace*{0.1in}
\noindent
{\bf (b)} 
$\opt^{\mami}=\frac{\opt_2}{2}$. 
\end{lemma}

\begin{proof}~\\
\noindent
{\bf (a)} 
Since only an edge with positive weight can increase the modularity of a cluster,
it is easy to check that a cluster with $y$ nodes can have a modularity value of at most $\frac{y\,d_{\max}}{4\,m}$.

\vspace*{0.1in}
\noindent
{\bf (b)} 
Consider an optimal clustering $\cS=\big\{V_1,V_2,\dots,V_k\big\}$ with a {\em minimum} number $k$ of clusters
such that $\opt^{\mami}=\M^{\mami}(\cS)=\min_{1\leq i\leq k}\big\{\M(V_i)\big\}>0$.
First, consider the case when $k>3$.
We will show that for some non-empty subset $T$ of $\{V_1,V_2,\ldots,V_k\}$
we must have $\M(\cup_{V_j\in T}V_j)\geq\M^{\mami}(\cS)$; this contradicts the minimality of $k$ in our choice of 
of the optimal cluster.
Note that $\M(\cS)=\sum_{i=1}^k \M(V_i)\geq k\cdot \M^{\mami}(\cS)$.
We will make use of Equation~\eqref{eq:11} of modularity of a cluster.
Let 
$\M(\widetilde{\cS})=
\frac{1}{2m} \left(
\sum_{\substack{u\in V_i,\, v\in V_j \\ i\neq j}}
\left(
a_{u,v}-\frac{d_u d_v}{2m}
\right)
\right)$.
Then, $\M(\widetilde{\cS})=-\M(\cS)$.
Consider a subset $T$ obtained by randomly and uniformly selecting each $V_i$ with a probability of $\nicefrac{1}{2}$.
Note that each pair of nodes $u$ and $v$ belonging to the same cluster is selected with a probability of
$\nicefrac{1}{2}$, whereas each pair of nodes belonging to different clusters is selected with a probability of $\nicefrac{1}{4}$.
Thus, 
\begin{gather*}
\expect\Big[\,\M(\cup_{V_j\in T}V_j)\,\Big]= \frac{\M(\cS)}{2}+\frac{\M(\widetilde{\cS})}{4}=\frac{\M(\cS)}{4}
\\
\geq \left(\frac{k}{4}\right)\M^{\mami}(\cS)\geq\M^{\mami}(\cS)
\end{gather*}
and therefore there exists such a subset $T$ with the properties as claimed.

Otherwise, consider the case when $k=3$. 
Let 
$\M_{i,j}
=
\frac{
\sum_{\substack{u\in V_i \\ v\in V_j}}
\left(
a_{u,v}-\frac{d_u d_v}{2m}
\right)
}
{2m}$
for $i<j$. 
Without loss of generality, let $\M(V_1)=a$, $\M(V_2)=a+b$ and $\M(V_3)=a+c$ for some $a>0$ and $b\geq c\geq 0$; thus,
$\M^{\mami}(\cS)=a$.
Consider the three $2$-clusterings of $G$: $C_1=\big(V_1\cup V_2,V_3\big)$, $C_2=\big(V_2\cup V_3, V_1\big)$ and $C_3=\big(V_1\cup V_3,V_2\big)$. 
Since none of these three $2$-clusterings should be an optimal solution, we must have 
\begin{multline*}
\M^{\mami}(C_1)-\M^{\mami}(\cS) < 0 
\\
\equiv
\min \big\{ 2a+b+\M_{1,2},\,a+c\big\} < a
\equiv
\M_{1,2}<-(a+b)
\end{multline*}
\begin{multline*}
\M^{\mami}(C_2)-\M^{\mami}(\cS) < 0 
\\
\equiv
\min \big\{ 2a+b+c+\M_{2,3},\,a\big\} < a
\equiv
M_{2,3}<-(a+b+c)
\end{multline*}
\begin{multline*}
\M^{\mami}(C_3)-\M^{\mami}(\cS) < 0 
\\
\equiv
\min \big\{ 2a+c+\M_{1,3},\,a\big\} < a
\equiv
M_{1,3}<-(a+c)
\end{multline*}
Thus, we have $\M(V_1)+\M(V_2)+\M(V_3)=3a+b+c=-\M_{1,2}-\M_{2,3}-\M_{1,3}>3a+2b+2c$
which implies $b+c<0$, contradicting $b\geq c\geq 0$.

Thus, we have shown there is an optimal solution for our $\mami$ objective with no more than two clusters.
Obviously, if $\opt^{\mami}>0$ then an optimal solution cannot consist of a single cluster.
Let $V_1,V_2$ be the two clusters in this case.
By Lemma~\ref{subgraph-selection}, we have 
$\M(V_1)=\M(V_2)$ which implies $\opt^{\mami}=\frac{\opt_2}{2}$.
\end{proof}

\subsection{Alternative Null Model: Erd\"{o}s-R\'{e}nyi Random Graphs}
\label{erg}

A theoretically appealing choice for alternative null models is 
the classical Erd\"{o}s-R\'{e}nyi random graph model $G(n,p)$,
namely 
each possible edge $\{u,v\}$ is selected in $G$
uniformly and randomly with a probability of $p$ for some fixed $0<p<1$.
To summarize, our results in this section show 
that the new modularity measure is precisely Newman's modularity measure 
on an appropriately defined regular graph, and thus 
our previous results on regular graphs can be applied to this case.

We will add the superscript ``{\sf ER}'' to differentiate the relevant quantities for
this objective from the usual summation objective discussed before, \EG, we will
use $\opt^{\mathsf{ER}}$ instead of $\opt$.
For simplicity, we consider the case of {\em unweighted graphs only}.
Let $G=(V,E)$ be the given unweighted input graph 
with $m=n\,\Delta$ number of edges.
Select $p=\frac{2\,\Delta}{n-1}$ such that the null model has 
the {\em same} number of edges {\em in expectation} as the given graph $G$.
Equation~\eqref{eq:11} then becomes 
\[
\M^{\mathrm{ER}}(C)=\frac{\displaystyle\sum_{u,v\in C}(a_{u,v}-p)}{2m}
\]
Let $n$ be sufficiently large such that 
$p\approx (2\Delta)/n$.
It can then be seen that 
$\M^{\mathsf{ER}}(C)$ is precisely the same 
as 
$\M(C)$ on a $(2\Delta)$-regular graph.
Thus, our previous results on regular graphs can be generalized to this case in the following manner:
\begin{itemize}
\item
Computing $\opt^{\mathsf{ER}}$ is $\NP$-complete for graphs with $\Delta\geq 18$.

\item
If $\Delta<\frac{n}{4\ln n}$ then the problem admits a $O(\log\Delta)$-approximation.
\end{itemize}

\vspace*{0.1in}
\noindent
{\bf Acknowledgements}
We thank Mario Szegedy for suggestion to investigate the problem for
dense graphs and other useful discussions, 
Geetha Jagannathan and Alantha Newman for useful discussions, and 
Mark Newman for explaining the significance of negative self-loops
in his modularity measure 
and pointing out references~\cite{KN09,LN08,N04}.

\end{document}